\documentclass[11pt]{article}

\usepackage[margin=1.25in]{geometry}
\usepackage[pdftex]{graphicx}
\usepackage[update,prepend]{epstopdf}
\usepackage[font=small,format=hang,justification=justified,labelfont=bf]{caption}
\usepackage{subcaption}
\usepackage[round,semicolon]{natbib}
\setcitestyle{notesep={; },round,aysep={},yysep={,}}
\usepackage{amsmath}
\usepackage{amsthm}
\usepackage{amsfonts}
\usepackage{amssymb}
\usepackage{dsfont}
\usepackage{bold-extra}
\usepackage{titlesec}
\usepackage[title]{appendix}
\usepackage{cases}
\usepackage{enumitem}
\setlist[enumerate]{leftmargin=*}
\usepackage{verbatim}
\usepackage{color}
\usepackage{mathtools}

\allowdisplaybreaks[4]

\usepackage{multicol,multirow,array}
\usepackage{booktabs}
\usepackage{wrapfig}
\usepackage{setspace}
\usepackage{fancybox}
\usepackage{fancyvrb}
\usepackage{framed}
\usepackage[pdfstartview={XYZ null null 0.85}]{hyperref}
\usepackage{url}
\usepackage{calc} % for the width/height calculations

\newcommand*{\MyLaw}{\mathrm{law}}
\newcommand*{\eqlawU}{\ensuremath{\mathop{\overset{\MyLaw}{=}}}} % unscaled version
\newcommand*{\eqlaw}{\mathop{\overset{\MyLaw}{\resizebox{\widthof{\eqlawU}}{\heightof{=}}{=}}}}

\newcommand\independent{\protect\mathpalette{\protect\independenT}{\perp}}
\def\independenT#1#2{\mathrel{\rlap{$#1#2$}\mkern2mu{#1#2}}}

\DeclareMathOperator{\E}{\mathbb{E}}

\newcommand{\RR}{\mathbb{R}}

\makeatletter
\newsavebox\tboxa
\newsavebox\tboxb
\newlength\tdima
\newcommand*{\oversymb}{\mathpalette\@oversymb}
\newcommand*{\@oversymb}[2]{
    \sbox{\tboxa}{$\m@th#1\mathrm{#2}$}
    \setbox\tboxb\null
    \ht\tboxb\ht\tboxa
    \dp\tboxb\dp\tboxa
    \wd\tboxb\wd\tboxa
    \sbox{\tboxa}{$\m@th#1{#2}$}
    \setlength\tdima{\the\wd\tboxa}
    \addtolength\tdima{-\the\wd\tboxb}
    \sbox{\tboxb}{$\m@th#1\hskip\tdima\overline{\xusebox{\tboxb}}$}
    \rlap{\usebox\tboxb}{\usebox\tboxa}}
\newcommand*{\xusebox}[1]{\mathord{{\usebox{#1}}}}
\makeatother

\newtheorem{theorem}{Theorem}[section]
\newtheorem{proposition}[theorem]{Proposition}
\newtheorem{lemma}[theorem]{Lemma}

\newtheorem{remark}[theorem]{Remark}
\let\oldremark\remark
\renewcommand{\remark}{\oldremark\normalfont}

\numberwithin{equation}{section}

\date{}
\begin{document}

\title{Exponential integrability properties of Euler discretization\\ schemes for the Cox-Ingersoll-Ross process}

\author{\scshape{andrei cozma}\thanks{\footnotesize\scshape{Mathematical Institute, University of Oxford, Oxford, OX2 6GG, UK}\newline
\hspace*{1.8em}andrei.cozma@maths.ox.ac.uk, christoph.reisinger@maths.ox.ac.uk} \and \scshape{christoph reisinger\footnotemark[1]}}

\maketitle

%%%%%%%%%%%%%%%%%%%%%%%%%%%%%%%%%%%%%%%%%%%%%%%%%%%%%%%%%%%%%%%%%%%%%%%%%%%%%%
%%% Abstract and Keywords %%%%%%%%%%%%%%%%%%%%%%%%%%%%%%%%%%%%%%%%%%%%%%%%%%%%%%%%%%%%%%%%
%%%%%%%%%%%%%%%%%%%%%%%%%%%%%%%%%%%%%%%%%%%%%%%%%%%%%%%%%%%%%%%%%%%%%%%%%%%%%%
\begin{abstract}\noindent

We analyze exponential integrability properties of the Cox-Ingersoll-Ross (CIR) process and its Euler discretizations with various types of truncation and reflection at $0$. These properties play a key role in establishing the finiteness of moments and the strong convergence of numerical approximations for a class of stochastic differential equations arising in finance. We prove that both implicit and explicit Euler-Maruyama discretizations for the CIR process preserve the exponential integrability of the exact solution for a wide range of parameters, and find lower bounds on the explosion time.

\vspace{1em}\noindent

\vspace{.5em}\noindent
\end{abstract}

%%%%%%%%%%%%%%%%%%%%%%%%%%%%%%%%%%%%%%%%%%%%%%%%%%%%%%%%%%%%%%%%%%%%%%%%%%%%%%
%%% Section 1 %%%%%%%%%%%%%%%%%%%%%%%%%%%%%%%%%%%%%%%%%%%%%%%%%%%%%%%%%%%%%%%%%%%%%%
%%%%%%%%%%%%%%%%%%%%%%%%%%%%%%%%%%%%%%%%%%%%%%%%%%%%%%%%%%%%%%%%%%%%%%%%%%%%%%
\section{Introduction}\label{sec:intro}

The Cox-Ingersoll-Ross process was originally proposed in \citet{Cox:1985} for short-term interest rate modeling, and is the solution to the following stochastic differential equation (SDE):
\begin{equation}\label{eq1.1}
dy_{t} = k_{y}(\theta_{y}-y_{t})dt + \xi_{y}\sqrt{y_{t}}\,dW_{t}\hspace{.5pt},
\end{equation}
where $W=(W_{t})_{t\geq0}$ is a one-dimensional Brownian motion, whereas $y_{0}$, $k_{y}$, $\theta_{y}$ and $\xi_{y}$ are strictly positive real numbers. According to \citet{Karatzas:1991}, \eqref{eq1.1} admits a unique strong solution, which is strictly positive when the Feller condition is satisfied, i.e., when $2k_{y}\theta_{y}>\xi_{y}^{2}$. The desirable features of the CIR process under consideration, such as non-negativity and mean-reversion, make it very popular when modeling interest rates or variances, e.g., in Heston's stochastic volatility model \citep{Heston:1993}. The Feller condition is typically satisfied in practice in the former case, but often fails to hold in the latter case.

The conditional distribution of the CIR process is noncentral chi-squared and hence its increments can be simulated exactly. On the other hand, discretization schemes are usually preferred when the entire sample path of the CIR process has to be simulated, or when the process is part of a system of SDEs. For instance, when pricing path-dependent financial derivatives written on an underlying process $S=(S_{t})_{t\in[0,T]}$ modeled by a $d$-dimensional SDE, with CIR dynamics in one or more dimensions, we need to evaluate
\begin{equation}\label{eq1.2}
U = \E\left[f(\hspace{.25pt}S\hspace{.75pt})\right],
\end{equation}
where $f:([0,T],\RR^{d})\mapsto\RR$ is the discounted payoff. In particular, this class of SDEs contains the popular Heston model and extensions thereof, such as stochastic interest rates \citep{Grzelak:2011, Ahlip:2013} or stochastic-local volatility \citep{Stoep:2014}. However, one can rarely find an explicit formula for the quantity in \eqref{eq1.2}, in which case we approximate the solution to the SDE via a discretization scheme and employ Monte Carlo simulation methods \citep[see][]{Glasserman:2003}. Since we cannot use the standard Euler-Maruyama scheme to approximate $(y_{t})_{t\in[0,T]}$ defined in \eqref{eq1.1} because of the non-zero probability of the approximation process becoming negative, we set it equal to zero when it turns negative (absorption fix) or reflect it in the origin (reflection fix). An overview of the explicit Euler schemes considered thus far in the literature can be found in \citet{Lord:2010}. Alternatively, we can use an implicit scheme to discretize the CIR process.

Although weak convergence is important when estimating expectations of payoffs, strong convergence may be required for complex path-dependent derivatives and plays a crucial role in multilevel Monte Carlo methods \citep{Giles:2008}. An important step in deriving strong convergence is proving the finiteness of moments of order higher than one of the process and its approximation \citep{Higham:2002, Cozma:2015b}. In addition, for a number of stochastic volatility models, moments of order higher than one can explode in finite time \citep{Andersen:2007}. This, however, can cause serious problems in practice when valuing securities whose payoffs have super-linear growth, as is the case with some commonly traded fixed income contracts. For instance, the risk-neutral valuation of CMS swaps and caps or Eurodollar futures contracts involves the evaluation of the second moment \citep{Andersen:2007}. Therefore, moment explosions may lead to infinite prices of derivatives. The same issue can also be observed for Euler approximations of SDEs with super-linearly growing drift or diffusion coefficients \citep{Hutzenthaler:2011, Jentzen:2015}.

Hence, we need to examine the stability of moments of the actual and the approximated processes, a problem directly related to the exponential integrability of the CIR process and its discretization \citep{Cozma:2015b}. \citet{Hutzenthaler:2014} proved that a class of stopped increment-tamed Euler approximations for nonlinear systems of SDEs with locally Lipschitz drift and diffusion coefficients retain the exponential integrability of the exact solution under some mild assumptions. However, the diffusion coefficient in \eqref{eq1.1} is not locally Lipschitz, so their analysis does not apply to the present work. \citet{Cozma:2015b} derived the exponential integrability of full
truncation Euler approximations \citep{Lord:2010} for the CIR process up to a critical time. In the present work, we first extend the aforementioned result to more general exponential functionals of the CIR process, and then, we prove that the drift-implicit and a number of explicit Euler discretizations for the CIR process preserve exponential integrability properties.

The paper is structured as follows. In Section~\ref{sec:scheme}, we discuss the discretization schemes and their strong convergence. In Section~\ref{sec:integrability}, we deduce the uniform exponential integrability of functionals of the CIR process and its explicit and implicit Euler discretizations. Section~\ref{sec:Heston} examines moment stability for a particular model (the Heston model) in more detail. Finally, Section~\ref{sec:conclusion} summarizes the results and outlines possible future work.

%%%%%%%%%%%%%%%%%%%%%%%%%%%%%%%%%%%%%%%%%%%%%%%%%%%%%%%%%%%%%%%%%%%%%%%%%%%%%%
%%% Section 2 %%%%%%%%%%%%%%%%%%%%%%%%%%%%%%%%%%%%%%%%%%%%%%%%%%%%%%%%%%%%%%%%%%%%%%
%%%%%%%%%%%%%%%%%%%%%%%%%%%%%%%%%%%%%%%%%%%%%%%%%%%%%%%%%%%%%%%%%%%%%%%%%%%%%%
\section{The discretization schemes}\label{sec:scheme}

The classical Euler-Maruyama scheme does not preserve the non-negativity of the process and hence it is not well-defined when applied to \eqref{eq1.1} directly due to the square-root diffusion coefficient. A number of corrections have been proposed in the literature, by either setting the process equal to zero when it turns negative, or by reflecting it in the origin. Consider a uniform grid: $\delta t=T\hspace{-.5pt}/N$, $t_{n}=n\hspace{.5pt}\delta t,\; \forall\hspace{.5pt} n\in\{0,1,...,N\}$. The \textit{partial truncation Euler (PTE) scheme}
\begin{equation}\label{eq2.1}
\tilde{y}_{t_{n+1}} = \tilde{y}_{t_{n}} + k_{y}(\theta_{y}-\tilde{y}_{t_{n}})\delta t + \xi_{y}\sqrt{\tilde{y}_{t_{n}}^{+}}\,\delta W_{t_{n}},
\end{equation}
where $y^{+}=\max\left(0,y\right)$ and $\delta W_{t_{n}}=W_{t_{n+1}}-W_{t_{n}}$, was proposed in \citet{Deelstra:1998}, whereas the \textit{full truncation Euler (FTE) scheme}
\begin{equation}\label{eq2.2}
\tilde{y}_{t_{n+1}} = \tilde{y}_{t_{n}} + k_{y}(\theta_{y}-\tilde{y}_{t_{n}}^{+})\delta t + \xi_{y}\sqrt{\tilde{y}_{t_{n}}^{+}}\,\delta W_{t_{n}}
\end{equation}
was studied in \citet{Lord:2010}. The \textit{absorption (ABS) scheme} reads as
\begin{equation}\label{eq2.3}
\tilde{y}_{t_{n+1}} = \tilde{y}_{t_{n}}^{+} + k_{y}(\theta_{y}-\tilde{y}_{t_{n}}^{+})\delta t + \xi_{y}\sqrt{\tilde{y}_{t_{n}}^{+}}\,\delta W_{t_{n}}.
\end{equation}
For the schemes \eqref{eq2.1} -- \eqref{eq2.3}, the piecewise constant time-continuous interpolation is defined as $\hspace{-.5pt}\oversymb{\hspace{.5pt}Y}_{\hspace{-2.5pt}t}=\tilde{y}_{t_{n}}^{+}$, whenever $t \in [t_{n},t_{n+1})$. The \textit{reflection (REF) scheme}
\begin{equation}\label{eq2.4}
\tilde{y}_{t_{n+1}} = \tilde{y}_{t_{n}} + k_{y}(\theta_{y}-\tilde{y}_{t_{n}})\delta t + \xi_{y}\sqrt{|\tilde{y}_{t_{n}}|}\,\delta W_{t_{n}}
\end{equation}
was introduced in \citet{Higham:2005}, and we define $\hspace{-.5pt}\oversymb{\hspace{.5pt}Y}_{\hspace{-2.5pt}t}=|\tilde{y}_{t_{n}}|$, whenever $t \in [t_{n},t_{n+1})$. The \textit{symmetrized Euler (SYM) scheme}
\begin{equation}\label{eq2.5}
\tilde{y}_{t_{n+1}} = \big|\tilde{y}_{t_{n}} + k_{y}(\theta_{y}-\tilde{y}_{t_{n}})\delta t + \xi_{y}\sqrt{\tilde{y}_{t_{n}}}\,\delta W_{t_{n}}\big|
\end{equation}
was studied in \citet{Bossy:2007}, and we let $\hspace{-.5pt}\oversymb{\hspace{.5pt}Y}_{\hspace{-2.5pt}t}=\tilde{y}_{t_{n}}$, whenever $t \in [t_{n},t_{n+1})$. Finally, assuming that the Feller condition holds and applying It\^o's formula to $x_{t}=\sqrt{y_{t}}$ leads to
\begin{equation}\label{eq2.6}
dx_{t} = \big(\alpha x_{t}^{-1}+\beta x_{t}\big)dt + \gamma\hspace{1pt}dW_{t}\hspace{.5pt},
\end{equation}
where
\begin{equation}\label{eq2.7}
\alpha = \frac{4k_{y}\theta_{y}-\xi_{y}^{2}}{8}\hspace{1pt},\hspace{2.5pt} \beta = -\hspace{1pt}\frac{k_{y}}{2} \hspace{2pt}\text{ and }\hspace{2pt} \gamma = \frac{\xi_{y}}{2}\hspace{1pt}.
\end{equation}
The \textit{drift-implicit (square-root) Euler scheme}
\begin{equation}\label{eq2.8}
\tilde{x}_{t_{n+1}} = \tilde{x}_{t_{n}} + \big(\alpha\tilde{x}_{t_{n+1}}^{-1}+\beta\tilde{x}_{t_{n+1}}\big)\delta t + \gamma\hspace{1pt}\delta W_{t_{n}}
\end{equation}
was proposed in \citet{Alfonsi:2005} and later studied in \citet{Dereich:2012}. Because $\alpha,\gamma>0$ and $\beta<0$, \eqref{eq2.8} has the unique positive solution
\begin{equation}\label{eq2.9}
\tilde{x}_{t_{n+1}} = \frac{\tilde{x}_{t_{n}}+\gamma\hspace{1pt}\delta W_{t_{n}}}{2(1-\beta\delta t)} + \sqrt{\frac{(\tilde{x}_{t_{n}}+\gamma\hspace{1pt}\delta W_{t_{n}})^{2}}{4(1-\beta\delta t)^{2}}+\frac{\alpha\hspace{.5pt}\delta t}{1-\beta\delta t}}\hspace{1pt}.
\end{equation}
This method is also called the \textit{backward Euler-Maruyama (BEM) scheme} \citep{Szpruch:2014}. We let $\hspace{-.5pt}\oversymb{\hspace{.5pt}Y}_{\hspace{-2.5pt}t}=\tilde{x}_{t_{n}}^{2}$, whenever $t \in [t_{n},t_{n+1})$.

The classical convergence theory \citep{Kloeden:1999, Higham:2002} does not apply to the CIR process because the square-root diffusion coefficient is not Lipschitz. Consequently, alternative approaches have been employed by the authors to prove the strong or weak convergence of their particular discretization. Strong convergence, either without a rate or with a logarithmic rate, of the partial truncation, the full truncation, the reflection and the symmetrized Euler schemes was established in \citet{Deelstra:1998}, \citet{Alfonsi:2005}, \citet{Higham:2005} and \citet{Lord:2010}. Strong convergence of order $1/2$ of the symmetrized and the drift-implicit Euler schemes was proved in \citet{Berkaoui:2008} and \citet{Dereich:2012}, respectively, albeit under some very restrictive assumptions for the former. Recently, \citet{Alfonsi:2013} and \citet{Szpruch:2014} improved the rate of strong convergence of the drift-implicit Euler scheme to $1$. To the best of our knowledge, convergence properties of the absorption scheme are not treated in the literature.

%%%%%%%%%%%%%%%%%%%%%%%%%%%%%%%%%%%%%%%%%%%%%%%%%%%%%%%%%%%%%%%%%%%%%%%%%%%%%%
%%% Section 3 %%%%%%%%%%%%%%%%%%%%%%%%%%%%%%%%%%%%%%%%%%%%%%%%%%%%%%%%%%%%%%%%%%%%%%
%%%%%%%%%%%%%%%%%%%%%%%%%%%%%%%%%%%%%%%%%%%%%%%%%%%%%%%%%%%%%%%%%%%%%%%%%%%%%%
\section{Exponential integrability}\label{sec:integrability}

The goal of this paper is to establish the exponential integrability of functionals of the CIR process and its approximations. To this end, let $(\hspace{.5pt}\oversymb{\hspace{.5pt}Y}_{\hspace{-2.5pt}t})_{t\in[0,T]}$ be the piecewise constant time-continuous approximation of $(y_{t})_{t\in[0,T]}$ corresponding to one of the discretization schemes from \eqref{eq2.1} -- \eqref{eq2.5} or \eqref{eq2.9} and, for some $\lambda, \mu \in \RR$, define
\begin{align}
\Theta_{t} &\equiv \exp\bigg\{\lambda\int_{0}^{t}{y_{u}\hspace{1pt}du} + \mu\int_{0}^{t}{\sqrt{y_{u}}\,dW_{u}}\bigg\},\hspace{2pt} \forall\hspace{.5pt}t\in[0,T],\label{eq3.1} \\[4pt]
\hspace{1pt}\oversymb{\hspace{-1pt}\Theta\hspace{-1pt}}\hspace{1pt}_{t} &\equiv \exp\bigg\{\lambda\int_{0}^{t}{\hspace{-.5pt}\oversymb{\hspace{.5pt}Y}_{\hspace{-2.5pt}u}\hspace{1pt}du} + \mu\int_{0}^{t}{\sqrt{\hspace{-.5pt}\oversymb{\hspace{.5pt}Y}_{\hspace{-2.5pt}u}}\,dW_{u}}\bigg\},\hspace{2pt} \forall\hspace{.5pt}t\in[0,T],\label{eq3.2}
\end{align}
and
\begin{equation}\label{eq3.2'}
\Delta \equiv \lambda + \frac{1}{2}\hspace{1pt}\mu^{2}.
\end{equation}

\begin{lemma}\label{Lem3.1}
Independent of the discretization scheme $\hspace{-.5pt}\oversymb{\hspace{.5pt}Y}$ employed,
\begin{align}\label{eq3.3}
\sup_{t \in [0,T]}\E\big[\hspace{1.5pt}\oversymb{\hspace{-1pt}\Theta\hspace{-1pt}}\hspace{1pt}_{t}\big] =
\begin{dcases}
	\hspace{12pt} 1 &\text{if } \Delta \leq 0\hspace{.5pt}, \\[1pt]
	\hspace{2pt} \E\big[\hspace{1.5pt}\oversymb{\hspace{-1pt}\Theta\hspace{-1pt}}\hspace{1pt}_{T}\big] &\text{if } \Delta>0\hspace{.5pt}.
\end{dcases}
\end{align}
\end{lemma}
\begin{proof}
Let $\left\{\mathcal{G}_{t},\hspace{1.5pt} 0\leq t\leq T\right\}$ be the natural filtration generated by $W$ and $\E_{t}\!\big[\hspace{.5pt}\cdot\hspace{.5pt}\big]\equiv\E\big[\hspace{.5pt}\cdot\hspace{.5pt}|\mathcal{G}_{t}\big]$. Assuming that $t \in [t_{n}, t_{n+1}]$ and conditioning on the $\sigma$-algebra $\mathcal{G}_{t_{n}}$, we find that
\begin{equation}\label{eq3.4}
\E_{t_{n}}\!\big[\hspace{1pt}\oversymb{\hspace{-1pt}\Theta\hspace{-1pt}}\hspace{1pt}_{t}\big] = \exp\bigg\{\lambda\int_{0}^{t_{n}}{\hspace{-.5pt}\oversymb{\hspace{.5pt}Y}_{\hspace{-2.5pt}u}\hspace{1pt}du} + \mu\int_{0}^{t_{n}}{\hspace{-2pt}\sqrt{\hspace{-.5pt}\oversymb{\hspace{.5pt}Y}_{\hspace{-2.5pt}u}}\,dW_{u}}\bigg\}\exp\Big\{(t-t_{n})\Delta\hspace{.5pt}\oversymb{\hspace{.5pt}Y}_{\hspace{-2.5pt}t_{n}}\Big\}.
\end{equation}
Since $\hspace{-.5pt}\oversymb{\hspace{.5pt}Y}$ is always non-negative, if $\Delta\leq0$, then $\E_{t_{n}}\!\big[\hspace{1pt}\oversymb{\hspace{-1pt}\Theta\hspace{-1pt}}\hspace{1pt}_{t}\big] \leq \E_{t_{n}}\!\big[\hspace{1pt}\oversymb{\hspace{-1pt}\Theta\hspace{-1pt}}\hspace{1pt}_{t_{n}}\big]$. From the law of iterated expectations,
\begin{equation}\label{eq3.5}
\E\big[\hspace{1pt}\oversymb{\hspace{-1pt}\Theta\hspace{-1pt}}\hspace{1pt}_{t}\big] \leq \E\big[\hspace{1pt}\oversymb{\hspace{-1pt}\Theta\hspace{-1pt}}\hspace{1pt}_{t_{n}}\big] \leq \E\big[\hspace{1pt}\oversymb{\hspace{-1pt}\Theta\hspace{-1pt}}\hspace{1pt}_{t_{n-1}}\big] \leq \hdots \leq \E\big[\hspace{1pt}\oversymb{\hspace{-1pt}\Theta\hspace{-1pt}}\hspace{1pt}_{0}\big].
\end{equation}
On the other hand, if $\Delta>0$, then $\E_{t_{n}}\!\big[\hspace{1pt}\oversymb{\hspace{-1pt}\Theta\hspace{-1pt}}\hspace{1pt}_{t}\big] \leq \E_{t_{n}}\!\big[\hspace{1pt}\oversymb{\hspace{-1pt}\Theta\hspace{-1pt}}\hspace{1pt}_{t_{n+1}}\big]$, so
\begin{equation}\label{eq3.6}
\E\big[\hspace{1pt}\oversymb{\hspace{-1pt}\Theta\hspace{-1pt}}\hspace{1pt}_{t}\big] \leq \E\big[\hspace{1pt}\oversymb{\hspace{-1pt}\Theta\hspace{-1pt}}\hspace{1pt}_{t_{n+1}}\big] \leq \E\big[\hspace{1pt}\oversymb{\hspace{-1pt}\Theta\hspace{-1pt}}\hspace{1pt}_{t_{n+2}}\big] \leq \hdots \leq \E\big[\hspace{1pt}\oversymb{\hspace{-1pt}\Theta\hspace{-1pt}}\hspace{1pt}_{T}\big],
\end{equation}
which concludes the proof.
\end{proof}

%%% Subsection 3.1 %%%
\subsection{The Cox-Ingersoll-Ross process}\label{subsec:CIR}

The exponential integrability result below is an extension of Proposition 3.1 in \citet{Andersen:2007} to the case $\Delta\leq0$. When $\Delta>0$, we can derive Proposition \ref{Prop3.2} directly from \citet{Andersen:2007}, by writing the moments of the Heston model in terms of $\E\!\big[\Theta_{T}\big]$. Our proof takes a fairly different approach and is included for completeness.

\begin{proposition}\label{Prop3.2}
The first moment of the exponential functional of the CIR process defined in \eqref{eq3.1} is uniformly bounded up to a critical time $T^{*}$, i.e.,
\begin{equation}\label{eq3.7}
\sup_{t\in[0,T]}\E\big[\Theta_{t}\big]<\infty,\hspace{6pt} \forall\hspace{1pt}T<T^{*},
\end{equation}
and
\begin{equation}\label{eq3.7'}
\E\big[\Theta_{T}\big]=\infty,\hspace{6pt} \forall\hspace{1pt}T\geq T^{*}.
\end{equation}
If $\Delta\leq0$, then $T^{*}=\infty$, whereas if $\Delta>0$, then $T^{*}$ is given below:
\begin{enumerate}
\item{When $k_{y}<\xi_{y}(\mu-\sqrt{2\Delta})$,
\begin{equation}\label{eq3.8}
T^{*} = \frac{1}{\nu}\log\left(\frac{\mu\xi_{y}-k_{y}+\nu}{\mu\xi_{y}-k_{y}-\nu}\right),\hspace{6pt} \nu = \sqrt{(\mu\xi_{y}-k_{y})^{2}-2\xi_{y}^{2}\Delta}\hspace{1pt}.
\end{equation}}
\item{When $k_{y}=\xi_{y}(\mu-\sqrt{2\Delta})$,
\begin{equation}\label{eq3.9}
T^{*} = \frac{2}{\mu\xi_{y}-k_{y}}\hspace{1pt}.
\end{equation}}
\item{When $\xi_{y}(\mu-\sqrt{2\Delta})<k_{y}<\xi_{y}(\mu+\sqrt{2\Delta})$,
\begin{equation}\label{eq3.10}
T^{*} = \frac{2}{\hat{\nu}}\left[\frac{\pi}{2}-\arctan\left(\frac{\mu\xi_{y}-k_{y}}{\hat{\nu}}\right)\right],\hspace{6pt} \hat{\nu}=\sqrt{2\xi_{y}^{2}\Delta-(\mu\xi_{y}-k_{y})^{2}}\hspace{1pt}.
\end{equation}}
\item{When $k_{y}\geq\xi_{y}(\mu+\sqrt{2\Delta})$,
\begin{equation}\label{eq3.11}
T^{*} = \infty\hspace{.5pt}.
\end{equation}}
\end{enumerate}
\end{proposition}
\begin{proof}
Since
\begin{equation}\label{eq3.12}
\int_{0}^{t}{\sqrt{y_{u}}\,dW_{u}} = \frac{1}{\xi_{y}}\hspace{1pt}y_{t} - \frac{1}{\xi_{y}}\big(y_{0}+k_{y}\theta_{y}t\big) + \frac{k_{y}}{\xi_{y}}\int_{0}^{t}{y_{u}\hspace{1pt}du}
\end{equation}
from \eqref{eq1.1},
\begin{equation}\label{eq3.13}
\Theta_{t} = \exp\left\{-\big(y_{0}+k_{y}\theta_{y}t\big)\hat{\lambda}\right\}\exp\bigg\{\hat{\lambda}y_{t} + \hat{\mu}\int_{0}^{t}{y_{u}\hspace{1pt}du}\bigg\},
\end{equation}
where
\begin{equation}\label{eq3.14}
\hat{\lambda} = \frac{\mu}{\xi_{y}}\hspace{1pt},\hspace{6pt} \hat{\mu} = \lambda + \frac{\mu k_{y}}{\xi_{y}}\hspace{1pt}.
\end{equation}
Since the first exponential on the right-hand side of \eqref{eq3.13} is a continuous function of time and since $(y_{t})_{t\in[0,T]}$ is a stationary Markov process, it suffices to prove the finiteness of the supremum over $\tau$ of
\begin{align}
F(\tau,y) &\equiv \E_{t}\!\left[\exp\bigg\{\hat{\lambda}y_{T} + \hat{\mu}\int_{t}^{T}{y_{u}\hspace{1pt}du}\bigg\}\hspace{1pt}\Big|\hspace{2pt}y_{t}=y\right]\label{eq3.15} \\[4pt]
&= \E\left[\exp\bigg\{\hat{\lambda}y_{\tau} + \hat{\mu}\int_{0}^{\tau}{y_{u}\hspace{1pt}du}\bigg\}\hspace{1pt}\Big|\hspace{2pt}y_{0}=y\right],\label{eq3.16}
\end{align}
where $\tau=T-t$ and $F(0,y)=\exp\{\hat{\lambda}y\}$. The process $M=(M_{t})_{t\in[0,T]}$ defined by
\begin{equation}\label{eq3.17}
M_{t} \equiv F(\tau,y)\exp\bigg\{\hat{\mu}\int_{0}^{t}{y_{u}\hspace{1pt}du}\bigg\}
\end{equation}
is a martingale. Applying It\^o's formula to the right-hand side and setting the resulting drift term equal to zero, we find a PDE for $F(\tau,y)$:
\begin{equation}\label{eq3.18}
-\partial_{\tau}F + k_{y}(\theta_{y}-y)\partial_{y}F + \frac{1}{2}\hspace{1pt}\xi_{y}^{2}y\partial_{yy}F + \hat{\mu}yF = 0.
\end{equation}
Suppose that the solution to the PDE is of the form
\begin{equation}\label{eq3.19}
F(\tau,y) = \exp\big\{\hat{\lambda}y + k_{y}\theta_{y}\hat{\lambda}\hspace{.5pt}\tau + G(\tau)y + H(\tau)\big\},
\end{equation}
with
\begin{equation}\label{eq3.20}
G(0)=0 \hspace{3pt}\text{ and }\hspace{3pt} H(0)=0.
\end{equation}
Substituting back into \eqref{eq3.18} with \eqref{eq3.19} results in a system of Riccati ODEs:
\begin{align}
\partial_{\tau}G(\tau) &= a\hspace{.5pt}G(\tau)^{2} + b\hspace{.5pt}G(\tau) + c,\label{eq3.21} \\[1pt]
\partial_{\tau}H(\tau) &= k_{y}\theta_{y}G(\tau),\label{eq3.22}
\end{align}
where
\begin{equation}\label{eq3.23}
a=\frac{1}{2}\hspace{1pt}\xi_{y}^{2},\hspace{5pt}
b=\hat{\lambda}\xi_{y}^{2}-k_{y}=\mu\xi_{y}-k_{y} \hspace{3pt}\text{ and }\hspace{3pt}
c=\frac{1}{2}\hspace{1pt}\hat{\lambda}^{2}\xi_{y}^{2}+\hat{\mu}-\hat{\lambda}k_{y}=\Delta.
\end{equation}
The conditions under which $G(\tau)$ -- and so $F(\tau,y)$ -- blows up are connected to the position of the roots of the polynomial $f(x)=ax^{2}+bx+c$, i.e.,
\begin{equation}\label{eq3.24}
x_{1,2} = -\frac{1}{\xi_{y}^{2}}(\mu\xi_{y}-k_{y}) \pm \frac{1}{\xi_{y}^{2}}\sqrt{(\mu\xi_{y}-k_{y})^{2}-2\xi_{y}^{2}\Delta}\hspace{1pt}.
\end{equation}
First of all, if $\Delta=0$, then zero is a root of the polynomial. Since $f(x)$ is locally Lipschitz, we conclude that the ODE \eqref{eq3.21} with the initial condition $G(0)=0$ has a unique solution, namely $G(\tau)=0$, $\forall\hspace{.5pt}\tau\geq0$. From \eqref{eq3.22}, we find that $H(\tau)=0$, $\forall\hspace{.5pt}\tau\geq0$. Substituting back into \eqref{eq3.19} and making use of \eqref{eq3.16}, we deduce that $\E\big[\Theta_{t}\big]=1$, $\forall\hspace{.5pt}t\geq0$. Alternatively, note that $(\Theta_{t})_{t\in[0,T]}$ is a Dol\'eans exponential and a true martingale \citep[see, for instance,][]{Cheridito:2007}. Therefore, when $\Delta=0$, $\sup_{t\in[0,T]}\E\big[\Theta_{t}\big]=1$ for all $T\geq0$.

Second, if $\Delta\neq0$, employing Proposition 3.3 in \citet{Liberty:2011}, one can easily show the finiteness of $G(\tau)$, and hence of $H(\tau)$ and $F(\tau,y)$, for all $\tau<T^{*}$, and the explosion of $G(\tau)$, for all $\tau\geq T^{*}$. Therefore, $F(\tau,y)$ is continuous and finite on $[0,T]$, for all $T<T^{*}$, so its supremum over the time interval is finite by the boundedness theorem, which concludes the proof.
\end{proof}

%%% Subsection 3.2 %%%
\subsection{The drift-implicit Euler (BEM) scheme}\label{subsec:BEM}

Suppose that $2k_{y}\theta_{y}>\xi_{y}^{2}$ and let $\hspace{-.5pt}\oversymb{\hspace{.5pt}Y}_{\hspace{-2.5pt}t}=\tilde{x}_{t_{n}}^{2}$, $\forall\hspace{1pt}t\in[t_{n},t_{n+1})$, with $\tilde{x}$ defined in \eqref{eq2.9}. In order to derive the exponential integrability of the BEM scheme, we first prove an auxiliary result.

\begin{lemma}\label{Lem3.3}
Suppose that $\Delta>0$. Then we can find $\eta>1$ such that for all $\omega\in[0,1]$,
\begin{equation}\label{eq3.25}
2\eta^{2}\omega^{2}\Delta\gamma^{2}T^{2} + 2\eta\hspace{.5pt}\omega\gamma\mu\hspace{.5pt}T - \eta + 1 < 0,
\end{equation}
if and only if $T<T^{*}$, where $T^{*}$ is given below:
\begin{enumerate}
\item{When $\mu<0$ and $\lambda<\frac{3}{2}\hspace{1pt}\mu^{2}$,
\begin{equation}\label{eq3.26}
T^{*} = -\frac{2\mu}{\xi_{y}\Delta}\hspace{1pt}.
\end{equation}}
\item{When $\mu<0$ and $\lambda\geq\frac{3}{2}\hspace{1pt}\mu^{2}$, or when $\mu\geq0$,
\begin{equation}\label{eq3.27}
T^{*} = \frac{1}{\xi_{y}(\mu+\sqrt{2\Delta})}\hspace{1pt}.
\end{equation}}
\end{enumerate}
\end{lemma}
\begin{proof}
Fix any $\eta>1$ and define the polynomial
\begin{equation}\label{eq3.28}
f_{\eta}(\omega) = 2\omega^{2}\eta^{2}\Delta\gamma^{2}T^{2} + 2\omega\eta\gamma\mu\hspace{.5pt}T - (\eta-1),
\end{equation}
with two distinct real roots
\begin{equation}\label{eq3.29}
\omega_{1,2} = \frac{-\mu\pm\sqrt{\mu^{2}+2(\eta-1)\Delta}}{2\eta\Delta\gamma\hspace{.5pt}T}\hspace{1pt}.
\end{equation}
Since $\omega_{1}>0>\omega_{2}$, we deduce that $f_{\eta}([0,1])<0$ if and only if $f_{\eta}(1)<0$, i.e.,
\begin{equation}\label{eq3.30}
2\eta^{2}\Delta\gamma^{2}T^{2} - \eta(1-2\gamma\mu\hspace{.5pt}T) + 1 < 0.
\end{equation}
However, \eqref{eq3.30} holds for some $\eta>1$ if and only if the second-order polynomial in $\eta$ on the left-hand side has two distinct real roots, one of which is greater than one. Substituting~back into \eqref{eq3.30} with $\gamma=0.5\hspace{.5pt}\xi_{y}$ from \eqref{eq2.7}, we find the necessary and sufficient conditions:
\begin{equation}\label{eq3.31}
(1-\mu\xi_{y}T)^{2} > 2\Delta\xi_{y}^{2}T^{2} \hspace{4pt}\text{ and }\hspace{4pt} 1-\mu\xi_{y}T + \sqrt{(1-\mu\xi_{y}T)^{2}-2\Delta\xi_{y}^{2}T^{2}} > \Delta\xi_{y}^{2}T^{2}.
\end{equation}
Some straightforward calculations lead to an equivalent set of conditions:
\begin{equation}\label{eq3.32}
\xi_{y}T(\mu+\sqrt{2\Delta})<1,
\end{equation}
and
\begin{equation}\label{eq3.33}
2\Delta\xi_{y}T<\sqrt{\mu^{2}+4\Delta}-\mu \hspace{5pt}\text{ or }\hspace{4pt} \sqrt{\mu^{2}+4\Delta}-\mu\leq2\Delta\xi_{y}T<-4\mu.
\end{equation}
Henceforth, the conclusion follows relatively easily.
\end{proof}

\begin{proposition}\label{Prop3.4}
If $\Delta\leq0$ and $T\geq0$ or otherwise, if $\Delta>0$ and $T<T^{*}$, with $T^{*}$ from \eqref{eq3.26} -- \eqref{eq3.27}, then there exists $\delta_{T}>0$ so that for all $\delta t\in(0,\delta_{T})$, the first moment of the exponential functional from \eqref{eq3.2} of the BEM scheme is uniformly bounded, i.e.,
\begin{equation}\label{eq3.34}
\sup_{\delta t\in(0,\delta_{T})}\hspace{1.5pt}\sup_{t\in[0,T]}\E\big[\hspace{1pt}\oversymb{\hspace{-1pt}\Theta\hspace{-1pt}}\hspace{1pt}_{t}\big]<\infty.
\end{equation}
\end{proposition}
\begin{proof}
If $\Delta\leq0$, the conclusion follows from Lemma \ref{Lem3.1}. If $\Delta>0$ and $T<T^{*}$, we know from Lemma \ref{Lem3.3} that $\exists\hspace{1pt}\eta>1$ independent of $\delta t$ such that \eqref{eq3.25} holds for all $\omega\in[0,1]$. Fix any such $\eta$. We prove by induction on $0\leq m\leq N$ that for sufficiently small values of $\delta t$,
\begin{align}\label{eq3.35}
\E\big[\hspace{1pt}\oversymb{\hspace{-1pt}\Theta\hspace{-1pt}}\hspace{1pt}_{T}\big] &\leq
\E\bigg[\exp\bigg\{\mu\int_{0}^{t_{N-m}}{\hspace{-.3em}\sqrt{\hspace{-.5pt}\oversymb{\hspace{.5pt}Y}_{\hspace{-2.5pt}u}}\,dW_{u}} + \lambda\delta t \sum_{i=0}^{N-m-1}{\hspace{-.2em}\hspace{-.5pt}\oversymb{\hspace{.5pt}Y}_{\hspace{-2.5pt}t_{i}}} + \eta m\Delta\delta t\hspace{1pt}\hspace{-.5pt}\oversymb{\hspace{.5pt}Y}_{\hspace{-2.5pt}t_{N-m}}\bigg\}\bigg] \nonumber\\[3pt]
&\hspace{1em}\times \Big(1+2\eta\Delta\gamma^{2}T\delta t\Big)^{\hspace{-1pt}2m}\exp\Big\{\eta\alpha\Delta(\delta t)^{2}(m-1)m\Big\}.
\end{align}
Note that when $m=0$, we have equality. Let us assume that \eqref{eq3.35} holds for $0\leq m<N$ and prove the inductive step. Conditioning on $\mathcal{G}_{t_{N-m-1}}$, we obtain
\begin{align}\label{eq3.36}
\E\big[\hspace{1pt}\oversymb{\hspace{-1pt}\Theta\hspace{-1pt}}\hspace{1pt}_{T}\big] &\leq
\Big(1+2\eta\Delta\gamma^{2}T\delta t\Big)^{\hspace{-1pt}2m}\exp\Big\{\eta\alpha\Delta(\delta t)^{2}(m-1)m\Big\} \nonumber\\[3pt]
&\hspace{1em}\times \E\bigg[\exp\bigg\{\mu\int_{0}^{t_{N-m-1}}{\hspace{-.3em}\sqrt{\hspace{-.5pt}\oversymb{\hspace{.5pt}Y}_{\hspace{-2.5pt}u}}\,dW_{u}} + \lambda\delta t \sum_{i=0}^{N-m-1}{\hspace{-.2em}\hspace{-.5pt}\oversymb{\hspace{.5pt}Y}_{\hspace{-2.5pt}t_{i}}}\bigg\} \nonumber\\[1pt]
&\hspace{1em}\times \E_{t_{N-m-1}}\!\bigg[\exp\bigg\{\eta m\Delta\delta t\hspace{1pt}\hspace{-.5pt}\oversymb{\hspace{.5pt}Y}_{\hspace{-2.5pt}t_{N-m}} + \mu\sqrt{\hspace{-.5pt}\oversymb{\hspace{.5pt}Y}_{\hspace{-2.5pt}t_{N-m-1}}}\,\delta W_{t_{N-m-1}}\bigg\}\bigg]\bigg].
\end{align}
Define $x=\tilde{x}_{t_{N-m-1}}$ and note that if $Z\sim\mathcal{N}\left(0,1\right)$, then $\mathcal{G}_{t_{N-m-1}} \independent \delta W_{t_{N-m-1}} \eqlaw \sqrt{\delta t}\hspace{1pt}Z$. Let $\mathcal{I}$ be the conditional (inner) expectation in \eqref{eq3.36}, then
\begin{equation}\label{eq3.37}
\mathcal{I} = \E_{0,x}\!\left[\exp\left\{\eta m\Delta\delta t\hspace{1pt}\psi(Z)^{2} + \mu x\sqrt{\delta t}\hspace{1pt}Z\right\}\right],
\end{equation}
where
\begin{equation}\label{eq3.38}
\psi(z) = \frac{x+\gamma\sqrt{\delta t}\hspace{1pt}z}{2(1-\beta\delta t)} + \sqrt{\frac{(x+\gamma\sqrt{\delta t}\hspace{1pt}z)^{2}}{4(1-\beta\delta t)^{2}}+\frac{\alpha\hspace{.5pt}\delta t}{1-\beta\delta t}}\hspace{1pt}.
\end{equation}
On $\Big\{\omega:\hspace{1pt} x+\gamma\sqrt{\delta t}\hspace{1pt}Z(\omega) \leq 0\Big\}$, since $\beta<0$,
\begin{equation}\label{eq3.39}
\psi(Z)^{2} < \alpha\hspace{.5pt}\delta t\hspace{.5pt}.
\end{equation}
On $\Big\{\omega:\hspace{1pt} x+\gamma\sqrt{\delta t}\hspace{1pt}Z(\omega) > 0\Big\}$,
\begin{equation}\label{eq3.40}
\psi(Z)^{2} < \big(x+\gamma\sqrt{\delta t}\hspace{1pt}Z\big)^{2} + 2\alpha\hspace{.5pt}\delta t\hspace{.5pt}.
\end{equation}
Hence, if we let $z_{0}=-\frac{x}{\gamma\sqrt{\delta t}}\hspace{1pt}$,
\begin{align}\label{eq3.41}
\mathcal{I} &\leq \int_{-\infty}^{z_{0}}{\frac{1}{\sqrt{2\pi}}\exp\left\{-\frac{1}{2}z^{2} + \mu x\sqrt{\delta t}\hspace{1pt}z + \eta m\alpha\Delta(\delta t)^{2}\right\}dz} \nonumber\\[3pt]
&\hspace{1em} +\int_{z_{0}}^{\infty}{\frac{1}{\sqrt{2\pi}}\exp\left\{-\frac{1}{2}z^{2} + \mu x\sqrt{\delta t}\hspace{1pt}z + \eta m\Delta\delta t\Big[\big(x+\gamma\sqrt{\delta t}\hspace{1pt}z\big)^{2} + 2\alpha\hspace{.5pt}\delta t\Big]\right\}dz}\hspace{.5pt}.
\end{align}
Suppose that $\delta_{T}\leq\big(2\eta\Delta\gamma^{2}T\big)^{-1}$ and define
\begin{equation}\label{eq3.42}
a = \sqrt{1-2\eta m\Delta\gamma^{2}(\delta t)^{2}}\hspace{1pt},\hspace{3pt} b = \mu x\sqrt{\delta t} + 2\eta m\Delta\gamma x(\delta t)^{\frac{3}{2}},\hspace{3pt} c = \sqrt{\eta m\Delta\delta t(x^{2}+2\alpha\hspace{.5pt}\delta t)}\hspace{1pt}.
\end{equation}
Some straightforward calculations lead to the following upper bound:
\begin{equation}\label{eq3.43}
\mathcal{I} \leq \exp\left\{\frac{1}{2}\hspace{1pt}\mu^{2}x^{2}\delta t + \eta m\alpha\Delta(\delta t)^{2}\right\}\Phi\left(z_{1}\right) + \frac{1}{a}\exp\left\{c^{2} + \frac{b^{2}}{2a^{2}}\right\}\Big\{1-\Phi\left(z_{2}\right)\Big\},
\end{equation}
where $\Phi$ is the standard normal CDF and
\begin{equation}\label{eq3.44}
z_{1} = -\frac{x}{\gamma\sqrt{\delta t}} - \mu x\sqrt{\delta t}\hspace{1pt},\hspace{6pt} z_{2} = -\frac{ax}{\gamma\sqrt{\delta t}} - \frac{b}{a}\hspace{1pt}.
\end{equation}
Suppose that $\delta_{T}\leq\big(\gamma\max\big\{0^{+},-\mu\big\}\big)^{-1}$. Then $1+\mu\gamma\delta t>0$ and hence $z_{1}<0$. From \eqref{eq3.42},
\begin{equation}\label{eq3.45}
a\gamma\sqrt{\delta t}\hspace{1pt}(z_{1}-z_{2}) = (1-a)(\mu\gamma x\delta t-ax) + \gamma\sqrt{\delta t}\hspace{1pt}(b-\mu x\sqrt{\delta t}\hspace{1pt}) = (1-a)x(1+\mu\gamma\delta t).
\end{equation}
Therefore, since $x>0$ and $a\in(0,1]$, $z_{2}\leq z_{1}<0$. From \eqref{eq3.42}, one can also show that
\begin{equation}\label{eq3.46}
c^{2} + \frac{b^{2}}{2a^{2}} - \frac{1}{2}\hspace{1pt}\mu^{2}x^{2}\delta t - \eta m\alpha\Delta(\delta t)^{2} = \eta m\alpha\Delta(\delta t)^{2} + \frac{1}{a^{2}}\hspace{1pt}\eta m\Delta x^{2}\delta t(1+\mu\gamma\delta t)^{2} \geq 0.
\end{equation}
Hence,
\begin{equation}\label{eq3.47}
\mathcal{I} \leq \frac{1}{a}\exp\left\{c^{2} + \frac{b^{2}}{2a^{2}}\right\}\Big\{1+\Phi\left(z_{1}\right)-\Phi\left(z_{2}\right)\Big\}.
\end{equation}
Applying the mean value theorem to $\Phi\in C^{1}$, we can find $z \in \left[z_{2},z_{1}\right]$ such that
\begin{equation}\label{eq3.48}
\Phi(z_{1}) - \Phi(z_{2}) = \left(z_{1}-z_{2}\right)\phi(z) \leq \left(z_{1}-z_{2}\right)\phi\left(z_{1}\right).
\end{equation}
Hence, using \eqref{eq3.45},
\begin{equation}\label{eq3.49}
\Phi(z_{1}) - \Phi(z_{2}) \leq \frac{1-a}{a\sqrt{2\pi}}\hspace{1pt}\cdot\hspace{1pt}\frac{x(1+\mu\gamma\delta t)}{\gamma\sqrt{\delta t}}\exp\left\{-\frac{x^{2}(1+\mu\gamma\delta t)^{2}}{2\gamma^{2}\delta t}\right\}.
\end{equation}
Consider the function $g:(0,\infty)\mapsto\RR$ defined by $g(u) = u\hspace{.5pt}e^{-\frac{u^{2}}{2}}$. Then the global maximum is $e^{-\frac{1}{2}}$, and is achieved when $u=1$. We can thus bound the term on the right-hand side of \eqref{eq3.49} from above to get
\begin{equation}\label{eq3.50}
\Phi(z_{1}) - \Phi(z_{2}) \leq \frac{1-a}{a\sqrt{2\pi e}} \leq \frac{1}{a}-1\hspace{.5pt}.
\end{equation}
Plugging back into \eqref{eq3.47},
\begin{equation}\label{eq3.51}
\mathcal{I} \leq \frac{1}{a^{2}}\exp\left\{c^{2} + \frac{b^{2}}{2a^{2}}\right\}.
\end{equation}
Suppose that $\delta_{T}\leq\frac{\sqrt{5}\hspace{1pt}-1}{4}\hspace{1pt}\big(\eta\Delta\gamma^{2}T\big)^{-1}$, then $a(1+a)>1$ and so $a^{-1}\leq2-a^{2}$. Using \eqref{eq3.42},
\begin{equation}\label{eq3.52}
\frac{1}{a^{2}} \leq \Big(1+2\eta m\Delta\gamma^{2}(\delta t)^{2}\Big)^{2} < \Big(1+2\eta\Delta\gamma^{2}T\delta t\Big)^{2}.
\end{equation}
Combining \eqref{eq3.46}, \eqref{eq3.51} and \eqref{eq3.52}, we deduce that
\begin{equation}\label{eq3.53}
\mathcal{I} < \Big(1+2\eta\Delta\gamma^{2}T\delta t\Big)^{2}\exp\left\{2\eta m\alpha\Delta(\delta t)^{2} + \frac{1}{2}\hspace{1pt}\mu^{2}x^{2}\delta t + \frac{1}{a^{2}}\hspace{1pt}\eta m\Delta x^{2}\delta t(1+\mu\gamma\delta t)^{2}\right\}.
\end{equation}
Due to our choice of $\eta$, the second-order polynomial $f_{\eta}(\omega)$ defined in \eqref{eq3.28} is negative for all $\omega\in[0,1]$. In particular, $f_{\eta}$ attains its maximum on a closed, bounded interval. Hence, there exists $\omega_{0}\in[0,1]$ independent of $\delta t$ and $m$ so that $f_{\eta}(\omega)\leq f_{\eta}(\omega_{0})<0$ for all $\omega\in[0,1]$. Suppose that $\delta_{T}\leq -f_{\eta}(\omega_{0})\big(2\eta(\eta\Delta-\lambda)\gamma^{2}T\big)^{-1}$, then
\begin{equation}\label{eq3.54}
\big(1-a^{2}\big)\bigg(\eta-1+\frac{\mu^{2}}{2\Delta}\bigg) < -f_{\eta}(\omega_{0}) \leq -f_{\eta}(\omega),\hspace{5pt} \forall\hspace{.5pt}\omega\in[0,1].
\end{equation}
Applying the inequality with $\omega=\frac{m}{N}$ and using \eqref{eq3.28}, we get
\begin{equation}\label{eq3.55}
\big(1-a^{2}\big)\bigg(\eta-1+\frac{\mu^{2}}{2\Delta}\bigg) < \eta - 1 - 2\eta^{2}m^{2}\Delta\gamma^{2}(\delta t)^{2} - 2\eta m\mu\gamma\hspace{.5pt}\delta t\hspace{.5pt}.
\end{equation}
However, using \eqref{eq3.42}, we can rewrite it as
\begin{equation}\label{eq3.56}
\frac{1}{a^{2}}\hspace{1pt}\eta m(1+\mu\gamma\delta t)^{2} < \eta(m+1) - 1,
\end{equation}
so
\begin{equation}\label{eq3.57}
\mathcal{I} < \Big(1+2\eta\Delta\gamma^{2}T\delta t\Big)^{2}\exp\Big\{2\eta m\alpha\Delta(\delta t)^{2} - \lambda x^{2}\delta t + \eta(m+1)\Delta x^{2}\delta t\Big\}.
\end{equation}
Substituting back into \eqref{eq3.36} with this upper bound gives the inductive step. Finally, taking $m=N$ in \eqref{eq3.35} leads to
\begin{align}\label{eq3.58}
\E\big[\hspace{1pt}\oversymb{\hspace{-1pt}\Theta\hspace{-1pt}}\hspace{1pt}_{T}\big] &\leq
\left(1+\frac{2\eta\Delta\gamma^{2}T^{2}}{N}\right)^{\hspace{-1pt}2N}\exp\left\{\eta\alpha\Delta T^{2}\bigg(1-\frac{1}{N}\bigg) + \eta\Delta Ty_{0}\right\} \nonumber\\[3pt]
&< \exp\Big\{\eta\Delta T^{2}\big(\alpha + 4\gamma^{2}\big) + \eta\Delta Ty_{0}\Big\}.
\end{align}
The right-hand side is finite and independent of $\delta t$, whence the conclusion. On a side note, following the same line of proof with $\eta=1$ leads to suboptimal sufficient conditions, whereas with $\eta<1$ we fail to achieve the inductive step.
\end{proof}

%%% Subsection 3.3 %%%
\subsection{Explicit Euler schemes with absorption fixes}\label{subsec:absorption}

First of all, consider the FTE scheme and let $\hspace{-.5pt}\oversymb{\hspace{.5pt}Y}_{\hspace{-2.5pt}t}=\tilde{y}_{t_{n}}^{+}$, $\forall\hspace{.5pt}t\in[t_{n},t_{n+1})$, with $\tilde{y}$ from \eqref{eq2.2}.

\begin{lemma}\label{Lem3.5}
Suppose that $\Delta>0$. Then we can find $\eta\geq1$ such that for all $\omega\in[0,1]$,
\begin{equation}\label{eq3.59}
\eta^{2}\omega^{2}\xi_{y}^{2}\Delta T^{2} - 2\eta\hspace{.5pt}\omega(k_{y}-\mu\xi_{y})T - 2\eta + 2 \leq 0,
\end{equation}
if and only if $T\leq T^{*}$, where $T^{*}$ is given below:
\begin{enumerate}
\item{When $k_{y}\leq\xi_{y}(\mu+\sqrt{0.5\hspace{.5pt}\Delta})$,
\begin{equation}\label{eq3.60}
T^{*} = \frac{1}{\xi_{y}(\mu+\sqrt{2\Delta})-k_{y}}\hspace{1pt}.
\end{equation}}
\item{When $k_{y}>\xi_{y}(\mu+\sqrt{0.5\hspace{.5pt}\Delta})$,
\begin{equation}\label{eq3.61}
T^{*} = \frac{2(k_{y}-\mu\xi_{y})}{\xi_{y}^{2}\Delta}\hspace{1pt}.
\end{equation}}
\end{enumerate}
\end{lemma}
\begin{proof}
Fix any $\eta\geq1$ and define the polynomial
\begin{equation}\label{eq3.62}
f_{\eta}(\omega) = \omega^{2}\eta^{2}\xi_{y}^{2}\Delta T^{2} - 2\omega\eta(k_{y}-\mu\xi_{y})T - 2(\eta-1),
\end{equation}
with real roots
\begin{equation}\label{eq3.63}
\omega_{1,2} = \frac{k_{y}-\mu\xi_{y}\pm\sqrt{(k_{y}-\mu\xi_{y})^{2}+2(\eta-1)\xi_{y}^{2}\Delta}}{\eta\xi_{y}^{2}\Delta T}\hspace{1pt}.
\end{equation}
Since $\omega_{1}\geq0\geq\omega_{2}$, we know that $f_{\eta}([0,1])\leq0$ if and only if $f_{\eta}(1)\leq0$, i.e.,
\begin{equation}\label{eq3.64}
\eta^{2}\xi_{y}^{2}\Delta T^{2} - 2\eta\big[1+(k_{y}-\mu\xi_{y})T\big] + 2 \leq 0.
\end{equation}
However, \eqref{eq3.64} holds for some $\eta\geq1$ if and only if the second-order polynomial in $\eta$ on the left-hand side has a real root greater or equal to one. Therefore, we find the necessary and sufficient conditions:
\begin{equation}\label{eq3.65}
\big[\xi_{y}(\mu+\sqrt{2\Delta})-k_{y}\big]T \leq 1,
\end{equation}
and
\begin{equation}\label{eq3.66}
2\Delta\xi_{y}^{2}T \leq k_{y}-\mu\xi_{y}+\sqrt{(k_{y}-\mu\xi_{y})^{2}+4\xi_{y}^{2}\Delta}
\end{equation}
or
\begin{equation}\label{eq3.67}
k_{y}-\mu\xi_{y}+\sqrt{(k_{y}-\mu\xi_{y})^{2}+4\xi_{y}^{2}\Delta} < 2\Delta\xi_{y}^{2}T \leq 4(k_{y}-\mu\xi_{y}).
\end{equation}
However, it is easy to see that conditions \eqref{eq3.65} -- \eqref{eq3.67} are equivalent to \eqref{eq3.60} -- \eqref{eq3.61}.
\end{proof}

The following result is an extension of Proposition 3.3 in \citet{Cozma:2015b}.

\begin{proposition}\label{Prop3.6}
If $\Delta\leq0$ and $T\geq0$ or otherwise, if $\Delta>0$ and $T\leq T^{*}$, with $T^{*}$ from \eqref{eq3.60} -- \eqref{eq3.61}, then there exists $\delta_{T}>0$ such that for all $\delta t\in(0,\delta_{T})$, the first moment of the exponential functional from \eqref{eq3.2} of the FTE scheme is uniformly bounded, i.e.,
\begin{equation}\label{eq3.68}
\sup_{\delta t\in(0,\delta_{T})}\hspace{1.5pt}\sup_{t\in[0,T]}\E\big[\hspace{1pt}\oversymb{\hspace{-1pt}\Theta\hspace{-1pt}}\hspace{1pt}_{t}\big]<\infty\hspace{.5pt}.
\end{equation}
\end{proposition}
\begin{proof}
If $\Delta\leq0$, this is a consequence of Lemma \ref{Lem3.1}. If $\Delta>0$ and $T\leq T^{*}$, we know from Lemma \ref{Lem3.5} that $\exists\hspace{1pt}\eta\geq1$ independent of $\delta t$ such that \eqref{eq3.59} holds for all $\omega\in[0,1]$. Fix any such $\eta$. We prove by induction on $0\leq m\leq N$ that for sufficiently small values of $\delta t$,
\begin{align}\label{eq3.69}
\E\big[\hspace{1pt}\oversymb{\hspace{-1pt}\Theta\hspace{-1pt}}\hspace{1pt}_{T}\big] &\leq
\E\bigg[\exp\bigg\{\mu\int_{0}^{t_{N-m}}{\hspace{-.3em}\sqrt{\hspace{-.5pt}\oversymb{\hspace{.5pt}Y}_{\hspace{-2.5pt}u}}\,dW_{u}} + \lambda\delta t \sum_{i=0}^{N-m-1}{\hspace{-.2em}\hspace{-.5pt}\oversymb{\hspace{.5pt}Y}_{\hspace{-2.5pt}t_{i}}} + \eta m\Delta\delta t\hspace{1pt}\hspace{-.5pt}\oversymb{\hspace{.5pt}Y}_{\hspace{-2.5pt}t_{N-m}}\bigg\}\bigg] \nonumber\\[3pt]
&\hspace{1em}\times \exp\Big\{0.5\hspace{.5pt}\eta\big(k_{y}\theta_{y} + \nu_{y}\xi_{y}\big)\Delta(\delta t)^{2}(m-1)m\Big\},
\end{align}
where
\begin{equation}\label{eq3.70}
\nu_{y} = \sqrt{\frac{1}{2\pi}\hspace{1pt}\xi_{y}^{2}+\frac{1}{2\pi}\hspace{1pt}\sqrt{\xi_{y}^{4}+2k_{y}^{2}\theta_{y}^{2}}}\hspace{1pt}.
\end{equation}
Note that when $m=0$, we have equality. Let us assume that \eqref{eq3.69} holds for $0 \leq m < N$ and prove the inductive step. Conditioning on $\mathcal{G}_{t_{N-m-1}}$, we obtain
\begin{align}\label{eq3.71}
\E\big[\hspace{1pt}\oversymb{\hspace{-1pt}\Theta\hspace{-1pt}}\hspace{1pt}_{T}\big] &\leq
\exp\Big\{0.5\hspace{.5pt}\eta\big(k_{y}\theta_{y} + \nu_{y}\xi_{y}\big)\Delta(\delta t)^{2}(m-1)m\Big\} \nonumber\\[2pt]
&\hspace{1em}\times \E\bigg[\exp\bigg\{\mu\int_{0}^{t_{N-m-1}}{\hspace{-.3em}\sqrt{\hspace{-.5pt}\oversymb{\hspace{.5pt}Y}_{\hspace{-2.5pt}u}}\,dW_{u}} + \lambda\delta t \sum_{i=0}^{N-m-1}{\hspace{-.2em}\hspace{-.5pt}\oversymb{\hspace{.5pt}Y}_{\hspace{-2.5pt}t_{i}}}\bigg\} \nonumber\\[1pt]
&\hspace{1em}\times \E_{t_{N-m-1}}\!\bigg[\exp\bigg\{\eta m\Delta\delta t\hspace{1pt}\hspace{-.5pt}\oversymb{\hspace{.5pt}Y}_{\hspace{-2.5pt}t_{N-m}} + \mu\sqrt{\hspace{-.5pt}\oversymb{\hspace{.5pt}Y}_{\hspace{-2.5pt}t_{N-m-1}}}\,\delta W_{t_{N-m-1}}\bigg\}\bigg]\bigg].
\end{align}
Define $\tilde{x} = \tilde{y}_{t_{N-m-1}}$ and $x = \hspace{-.5pt}\oversymb{\hspace{.5pt}Y}_{\hspace{-2.5pt}t_{N-m-1}}$. If $Z\sim\mathcal{N}\left(0,1\right)$, then $\mathcal{G}_{t_{N-m-1}} \independent \delta W_{t_{N-m-1}} \eqlaw \sqrt{\delta t}\hspace{1pt}Z$. Let $\mathcal{I}$ be the conditional expectation in \eqref{eq3.71}, then
\begin{equation}\label{eq3.72}
\mathcal{I} \leq \E_{0,x}\!\left[\exp\left\{\eta m\Delta\delta t\hspace{.5pt}\max\!\Big[0,\hspace{1pt} x + k_{y}(\theta_{y} - x)\delta t + \xi_{y}\sqrt{x\delta t}\hspace{1pt}Z\Big] + \mu\sqrt{x\delta t}\hspace{1pt}Z\right\}\right].
\end{equation}
If $x=0$, then
\begin{equation}\label{eq3.73}
\mathcal{I} \leq \exp\Big\{\eta m k_{y}\theta_{y}\Delta(\delta t)^{2}\Big\}.
\end{equation}
If $x>0$ and
\begin{equation}\label{eq3.74}
z_{0}=-\frac{k_{y}\theta_{y}\delta t + (1-k_{y}\delta t)x}{\xi_{y}\sqrt{x\delta t}}\hspace{1pt},
\end{equation}
then
\begin{align}\label{eq3.75}
\mathcal{I} &\leq
\int_{z_{0}}^{\infty}{\frac{1}{\sqrt{2\pi}}\exp\left\{-\frac{1}{2}z^{2} + \mu\sqrt{x\delta t}\hspace{1pt}z + \eta m\Delta\delta t\Big[x + k_{y}(\theta_{y} - x)\delta t + \xi_{y}\sqrt{x\delta t}\hspace{1pt}z\Big]\right\}dz} \nonumber\\[3pt]
&\hspace{1em} + \int_{-\infty}^{z_{0}}{\frac{1}{\sqrt{2\pi}}\exp\left\{-\frac{1}{2}z^{2} + \mu\sqrt{x\delta t}\hspace{1pt}z\right\}dz}\hspace{.5pt}.
\end{align}
Suppose that $\delta_{T}\leq\max\big\{0^{+},k_{y}-\mu\xi_{y}\big\}^{-1}$ and define
\begin{equation}\label{eq3.76}
a = \eta m\big[1 - (k_{y}-\mu\xi_{y})\delta t\big] + \frac{1}{2}\hspace{1pt}\eta^{2}m^{2}\xi_{y}^{2}\Delta(\delta t)^{2}\geq 0.
\end{equation}
Some straightforward calculations lead to the following upper bound:
\begin{equation}\label{eq3.77}
\mathcal{I} \leq \exp\bigg\{\eta m k_{y}\theta_{y}\Delta(\delta t)^{2} + \frac{1}{2}\hspace{1pt}\mu^{2}x\delta t + a\Delta x\delta t\bigg\}\Big\{1 + \Phi(z_{1}) - \Phi(z_{2})\Big\},
\end{equation}
where
\begin{equation}\label{eq3.78}
z_{1} = z_{0} - \mu\sqrt{x\delta t} = -\frac{k_{y}\theta_{y}\delta t + \big[1 - (k_{y}-\mu\xi_{y})\delta t\big]x}{\xi_{y}\sqrt{x\delta t}}
\end{equation}
and
\begin{equation}\label{eq3.79}
z_{2} = z_{0} - \mu\sqrt{x\delta t} - \eta m\xi_{y}\Delta(\delta t)^{3/2}\sqrt{x}\hspace{1pt}.
\end{equation}
Clearly $z_{2} \leq z_{1} < 0$, so we can find $z \in \left[z_{2},z_{1}\right]$ such that
\begin{equation}\label{eq3.80}
\Phi(z_{1}) - \Phi(z_{2}) = \left(z_{1}-z_{2}\right)\phi(z) \leq \left(z_{1}-z_{2}\right)\phi\left(z_{1}\right).
\end{equation}
Hence, using \eqref{eq3.78} -- \eqref{eq3.80},
\begin{equation}\label{eq3.81}
\Phi(z_{1}) - \Phi(z_{2}) \leq 
\frac{1}{\sqrt{2\pi}}\hspace{1pt}\eta m\xi_{y}\Delta(\delta t)^{3/2}g(x),
\end{equation}
where $g:(0,\infty)\mapsto\RR$ is defined by
\begin{equation}\label{eq3.82}
g(x) = \sqrt{x}\hspace{1pt}\exp\left\{-\frac{\big[k_{y}\theta_{y}\delta t + \left[1 - (k_{y}-\mu\xi_{y})\delta t\right]x\big]^2}{2\hspace{.5pt}\xi_{y}^{2}x\delta t}\right\}.
\end{equation}
Suppose that $\delta_{T}\leq\frac{\sqrt{2}-1}{\sqrt{2}}\max\big\{0^{+},k_{y}-\mu\xi_{y}\big\}^{-1}$. One can easily find the global maximum of the function, and hence an upper bound:
\begin{equation}\label{eq3.83}
g(x) < \nu_{y}\sqrt{2\pi\delta t}\hspace{1pt}.
\end{equation}
Substituting back into \eqref{eq3.77} with \eqref{eq3.81} and \eqref{eq3.83}, we get
\begin{equation}\label{eq3.84}
\mathcal{I} \leq \exp\bigg\{\eta m\big(k_{y}\theta_{y}+\nu_{y}\xi_{y}\big)\Delta(\delta t)^{2} + \frac{1}{2}\hspace{1pt}\mu^{2}x\delta t + a\Delta x\delta t\bigg\}.
\end{equation}
Note from \eqref{eq3.73} that this holds when $x=0$ as well. Applying \eqref{eq3.59} with $\omega=\frac{m}{N}$ leads to
\begin{equation}\label{eq3.85}
\eta^{2}m^{2}\xi_{y}^{2}\Delta(\delta t)^{2} - 2\eta m(k_{y}-\mu\xi_{y})\delta t - 2\eta + 2 \leq 0\hspace{.5pt}.
\end{equation}
Hence, from \eqref{eq3.76},
\begin{equation}\label{eq3.86}
a \leq \eta(m+1) - 1.
\end{equation}
Therefore,
\begin{equation}\label{eq3.87}
\mathcal{I} \leq \exp\bigg\{\eta m\big(k_{y}\theta_{y}+\nu_{y}\xi_{y}\big)\Delta(\delta t)^{2} - \lambda x\delta t + \eta(m+1)\Delta x\delta t\bigg\}.
\end{equation}
Substituting back into \eqref{eq3.71} with this upper bound gives the inductive step. Finally, taking $m=N$ in \eqref{eq3.69} leads to
\begin{equation}\label{eq3.88}
\E\big[\hspace{1pt}\oversymb{\hspace{-1pt}\Theta\hspace{-1pt}}\hspace{1pt}_{T}\big] < \exp\left\{\frac{1}{2}\hspace{1pt}\eta\Delta T^{2}\big(k_{y}\theta_{y}+\nu_{y}\xi_{y}\big) + \eta\Delta Ty_{0}\right\}.
\end{equation}
The right-hand side is finite and independent of $\delta t$, whence the conclusion. On a side note, if we follow the same line of proof with $\eta<1$, we fail to achieve the inductive step.
\end{proof}

Second, we consider the partial truncation and the absorption schemes and let $\hspace{-.5pt}\oversymb{\hspace{.5pt}Y}_{\hspace{-2.5pt}t}=\tilde{y}_{t_{n}}^{+}$, $\forall\hspace{.5pt}t\in[t_{n},t_{n+1})$, with $\tilde{y}$ defined in \eqref{eq2.1} and \eqref{eq2.3}, respectively.

\begin{proposition}\label{Prop3.7}
If $\Delta\leq0$ and $T\geq0$ or otherwise, if $\Delta>0$ and $T\leq T^{*}$, with $T^{*}$ from \eqref{eq3.60} -- \eqref{eq3.61}, then there exists $\delta_{T}>0$ such that for all $\delta t\in(0,\delta_{T})$, the first moments of the exponential functionals from \eqref{eq3.2} of the partial truncation and absorption schemes are uniformly bounded, i.e.,
\begin{equation}\label{eq3.89}
\sup_{\delta t\in(0,\delta_{T})}\hspace{1.5pt}\sup_{t\in[0,T]}\E\big[\hspace{1pt}\oversymb{\hspace{-1pt}\Theta\hspace{-1pt}}\hspace{1pt}_{t}\big]<\infty\hspace{.5pt}.
\end{equation}
\end{proposition}
\begin{proof}
We follow the argument of Proposition \ref{Prop3.6} closely and note that \eqref{eq3.72} holds for the partial truncation scheme if $\delta_{T}\leq k_{y}^{-1}$, while for the absorption scheme we have equality.
\end{proof}

%%% Subsection 3.4 %%%
\subsection{Explicit Euler schemes with reflection fixes}\label{subsec:reflection}

First, consider the reflection scheme and let $\hspace{-.5pt}\oversymb{\hspace{.5pt}Y}_{\hspace{-2.5pt}t}=|\tilde{y}_{t_{n}}|$, $\forall\hspace{.5pt}t\in[t_{n},t_{n+1})$, with $\tilde{y}$ from \eqref{eq2.4}.

\begin{proposition}\label{Prop3.8}
If $\Delta\leq0$ and $T\geq0$ or otherwise, if $\Delta>0$ and $T\leq T^{*}$, then there exists $\delta_{T}>0$ such that for all $\delta t\in(0,\delta_{T})$, the first moment of the exponential functional from \eqref{eq3.2} of the reflection scheme is uniformly bounded, i.e.,
\begin{equation}\label{eq3.90}
\sup_{\delta t\in(0,\delta_{T})}\hspace{1.5pt}\sup_{t\in[0,T]}\E\big[\hspace{1pt}\oversymb{\hspace{-1pt}\Theta\hspace{-1pt}}\hspace{1pt}_{t}\big]<\infty\hspace{.5pt},
\end{equation}
where $T^{*}$ is given below:
\begin{enumerate}
\item{When $k_{y}\leq\xi_{y}(|\mu|+\sqrt{0.5\hspace{.5pt}\Delta})$,
\begin{equation}\label{eq3.91}
T^{*} = \frac{1}{\xi_{y}(|\mu|+\sqrt{2\Delta})-k_{y}}\hspace{1pt}.
\end{equation}}
\item{When $k_{y}>\xi_{y}(|\mu|+\sqrt{0.5\hspace{.5pt}\Delta})$,
\begin{equation}\label{eq3.92}
T^{*} = \frac{2(k_{y}-|\mu|\xi_{y})}{\xi_{y}^{2}\Delta}\hspace{1pt}.
\end{equation}}
\end{enumerate}
\end{proposition}
\begin{proof}
If $\Delta\leq0$, this is a consequence of Lemma \ref{Lem3.1}. If $\Delta>0$ and $T\leq T^{*}$, we know from Lemma \ref{Lem3.5} that $\exists\hspace{1pt}\eta\geq1$ independent of $\delta t$ such that for all $\omega\in[0,1]$,
\begin{equation}\label{eq3.93}
\eta^{2}\omega^{2}\xi_{y}^{2}\Delta T^{2} - 2\eta\hspace{.5pt}\omega(k_{y}-|\mu|\xi_{y})T - 2\eta + 2 \leq 0.
\end{equation}
Fix any such $\eta$. Next, we prove by induction on $0\leq m\leq N$ that for sufficiently small~values of $\delta t$, we have
\begin{align}\label{eq3.94}
\E\big[\hspace{1pt}\oversymb{\hspace{-1pt}\Theta\hspace{-1pt}}\hspace{1pt}_{T}\big] &\leq
\E\bigg[\exp\bigg\{\mu\int_{0}^{t_{N-m}}{\hspace{-.3em}\sqrt{\hspace{-.5pt}\oversymb{\hspace{.5pt}Y}_{\hspace{-2.5pt}u}}\,dW_{u}} + \lambda\delta t \sum_{i=0}^{N-m-1}{\hspace{-.2em}\hspace{-.5pt}\oversymb{\hspace{.5pt}Y}_{\hspace{-2.5pt}t_{i}}} + \eta m\Delta\delta t\hspace{1pt}\hspace{-.5pt}\oversymb{\hspace{.5pt}Y}_{\hspace{-2.5pt}t_{N-m}}\bigg\}\bigg] \nonumber\\[3pt]
&\hspace{1em}\times \exp\Big\{\eta\big(0.5\hspace{.5pt}k_{y}\theta_{y} + \nu_{y}\xi_{y}\big)\Delta(\delta t)^{2}(m-1)m\Big\},
\end{align}
with $\nu_{y}$ from \eqref{eq3.70}. Note that when $m=0$, we have equality. Let us assume that \eqref{eq3.94} holds for $0 \leq m < N$ and prove the inductive step. Conditioning on $\mathcal{G}_{t_{N-m-1}}$, we obtain
\begin{align}\label{eq3.95}
\E\big[\hspace{1pt}\oversymb{\hspace{-1pt}\Theta\hspace{-1pt}}\hspace{1pt}_{T}\big] &\leq
\exp\Big\{\eta\big(0.5\hspace{.5pt}k_{y}\theta_{y} + \nu_{y}\xi_{y}\big)\Delta(\delta t)^{2}(m-1)m\Big\} \nonumber\\[2pt]
&\hspace{1em}\times \E\bigg[\exp\bigg\{\mu\int_{0}^{t_{N-m-1}}{\hspace{-.3em}\sqrt{\hspace{-.5pt}\oversymb{\hspace{.5pt}Y}_{\hspace{-2.5pt}u}}\,dW_{u}} + \lambda\delta t \sum_{i=0}^{N-m-1}{\hspace{-.2em}\hspace{-.5pt}\oversymb{\hspace{.5pt}Y}_{\hspace{-2.5pt}t_{i}}}\bigg\} \nonumber\\[1pt]
&\hspace{1em}\times \E_{t_{N-m-1}}\!\bigg[\exp\bigg\{\eta m\Delta\delta t\hspace{1pt}\hspace{-.5pt}\oversymb{\hspace{.5pt}Y}_{\hspace{-2.5pt}t_{N-m}} + \mu\sqrt{\hspace{-.5pt}\oversymb{\hspace{.5pt}Y}_{\hspace{-2.5pt}t_{N-m-1}}}\,\delta W_{t_{N-m-1}}\bigg\}\bigg]\bigg].
\end{align}
Define $\tilde{x} = \tilde{y}_{t_{N-m-1}}$ and $x = \hspace{-.5pt}\oversymb{\hspace{.5pt}Y}_{\hspace{-2.5pt}t_{N-m-1}}$. If $Z\sim\mathcal{N}\left(0,1\right)$, then $\mathcal{G}_{t_{N-m-1}} \independent \delta W_{t_{N-m-1}} \eqlaw \sqrt{\delta t}\hspace{1pt}Z$. Let $\mathcal{I}$ be the conditional expectation in \eqref{eq3.95}, then
\begin{equation}\label{eq3.96}
\mathcal{I} = \E_{0,\tilde{x}}\!\left[\exp\left\{\eta m\Delta\delta t\hspace{1pt}\big|\tilde{x} + k_{y}(\theta_{y} - \tilde{x})\delta t + \xi_{y}\sqrt{x\delta t}\hspace{1pt}Z\big| + \mu\sqrt{x\delta t}\hspace{1pt}Z\right\}\right].
\end{equation}
If $\tilde{x}=0$, then
\begin{equation}\label{eq3.97}
\mathcal{I} = \exp\Big\{\eta m k_{y}\theta_{y}\Delta(\delta t)^{2}\Big\}.
\end{equation}
If $\tilde{x}>0$ and
\begin{equation}\label{eq3.98}
z_{0}=-\frac{k_{y}\theta_{y}\delta t + (1-k_{y}\delta t)x}{\xi_{y}\sqrt{x\delta t}}\hspace{1pt},
\end{equation}
since $x=\tilde{x}$, we get
\begin{align}\label{eq3.99}
\mathcal{I} &=
\int_{-\infty}^{z_{0}}{\frac{1}{\sqrt{2\pi}}\exp\left\{-\frac{1}{2}z^{2} + \mu\sqrt{x\delta t}\hspace{1pt}z - \eta m\Delta\delta t\Big[x + k_{y}(\theta_{y} - x)\delta t + \xi_{y}\sqrt{x\delta t}\hspace{1pt}z\Big]\right\}dz} \nonumber\\[2pt]
&\hspace{.25em}+ \int_{z_{0}}^{\infty}{\frac{1}{\sqrt{2\pi}}\exp\left\{-\frac{1}{2}z^{2} + \mu\sqrt{x\delta t}\hspace{1pt}z + \eta m\Delta\delta t\Big[x + k_{y}(\theta_{y} - x)\delta t + \xi_{y}\sqrt{x\delta t}\hspace{1pt}z\Big]\right\}dz}\hspace{.25pt}.
\end{align}
Suppose that $\delta_{T}\leq\max\big\{0^{+},k_{y}-\mu\xi_{y}\big\}^{-1}$. Some straightforward calculations lead to
\begin{equation}\label{eq3.100}
\mathcal{I} \leq \exp\bigg\{\eta m k_{y}\theta_{y}\Delta(\delta t)^{2} + \frac{1}{2}\hspace{1pt}\mu^{2}x\delta t + a\Delta x\delta t\bigg\}\Big\{1 + \Phi(z_{1}) - \Phi(z_{2})\Big\},
\end{equation}
with $a$ defined in \eqref{eq3.76} and
\begin{equation}\label{eq3.101}
z_{1,2} = z_{0} - \big(\mu\mp\eta m\xi_{y}\Delta\delta t\big)\sqrt{x\delta t} = -\frac{k_{y}\theta_{y}\delta t + \big[1 - \big(k_{y}-\mu\xi_{y}\pm\eta m\xi_{y}^{2}\Delta\delta t\big)\delta t\big]x}{\xi_{y}\sqrt{x\delta t}}\hspace{1pt}.
\end{equation}
Suppose that $\delta_{T}\leq\max\big\{0^{+},k_{y}-\mu\xi_{y}+\eta\xi_{y}^{2}\Delta T\big\}^{-1}$, then $z_{2} \leq z_{1} < 0$ and hence we can find $z \in \left[z_{2},z_{1}\right]$ such that
\begin{equation}\label{eq3.103}
\Phi(z_{1}) - \Phi(z_{2}) = \left(z_{1}-z_{2}\right)\phi(z) \leq \left(z_{1}-z_{2}\right)\phi\left(z_{1}\right) = \sqrt{\frac{2}{\pi}}\hspace{1pt}\eta m\xi_{y}\Delta(\delta t)^{3/2}g(x),
\end{equation}
where $g:(0,\infty)\mapsto\RR$ is defined by
\begin{equation}\label{eq3.104}
g(x) = \sqrt{x}\hspace{1pt}\exp\left\{-\frac{\big[k_{y}\theta_{y}\delta t + \big[1 - \big(k_{y}-\mu\xi_{y}+\eta\xi_{y}^{2}\Delta T\big)\delta t\big]x\big]^2}{2\hspace{.5pt}\xi_{y}^{2}x\delta t}\right\}.
\end{equation}
Suppose that $\delta_{T}\leq\frac{\sqrt{2}-1}{\sqrt{2}}\max\big\{0^{+},k_{y}-\mu\xi_{y}+\eta\xi_{y}^{2}\Delta T\big\}^{-1}$. Proceeding as before, we can find an upper bound:
\begin{equation}\label{eq3.105}
g(x) < \nu_{y}\sqrt{2\pi\delta t}\hspace{1pt}.
\end{equation}
Substituting back into \eqref{eq3.100} with \eqref{eq3.103} and \eqref{eq3.105}, we get
\begin{equation}\label{eq3.106}
\mathcal{I} \leq \exp\bigg\{\eta m\big(k_{y}\theta_{y}+2\nu_{y}\xi_{y}\big)\Delta(\delta t)^{2} + \frac{1}{2}\hspace{1pt}\mu^{2}x\delta t + a\Delta x\delta t\bigg\}.
\end{equation}
If $\tilde{x}<0$ and
\begin{equation}\label{eq3.107}
z'_{0}=-\frac{k_{y}\theta_{y}\delta t - (1-k_{y}\delta t)x}{\xi_{y}\sqrt{x\delta t}}\hspace{1pt},
\end{equation}
since $x=-\tilde{x}$, we get
\begin{align}\label{eq3.108}
\mathcal{I} &=
\int_{-\infty}^{z'_{0}}{\frac{1}{\sqrt{2\pi}}\exp\left\{-\frac{1}{2}z^{2} + \mu\sqrt{x\delta t}\hspace{1pt}z - \eta m\Delta\delta t\Big[k_{y}(\theta_{y} + x)\delta t - x + \xi_{y}\sqrt{x\delta t}\hspace{1pt}z\Big]\right\}dz} \nonumber\\[2pt]
&\hspace{.25em}+ \int_{z'_{0}}^{\infty}{\frac{1}{\sqrt{2\pi}}\exp\left\{-\frac{1}{2}z^{2} + \mu\sqrt{x\delta t}\hspace{1pt}z + \eta m\Delta\delta t\Big[k_{y}(\theta_{y} + x)\delta t - x + \xi_{y}\sqrt{x\delta t}\hspace{1pt}z\Big]\right\}dz}\hspace{.25pt}.
\end{align}
Suppose that $\delta_{T}\leq\max\big\{0^{+},k_{y}+\mu\xi_{y}\big\}^{-1}$ and define
\begin{equation}\label{eq3.109}
b = \eta m\big[1 - (k_{y}+\mu\xi_{y})\delta t\big] + \frac{1}{2}\hspace{1pt}\eta^{2}m^{2}\xi_{y}^{2}\Delta(\delta t)^{2}\geq 0.
\end{equation}
Some straightforward calculations lead to the following upper bound:
\begin{equation}\label{eq3.110}
\mathcal{I} \leq \exp\bigg\{\eta m k_{y}\theta_{y}\Delta(\delta t)^{2} + \frac{1}{2}\hspace{1pt}\mu^{2}x\delta t + b\Delta x\delta t\bigg\}\Big\{1 + \Phi(z'_{1}) - \Phi(z'_{2})\Big\},
\end{equation}
where
\begin{equation}\label{eq3.111}
z'_{1,2} = z'_{0} - \big(\mu\mp\eta m\xi_{y}\Delta\delta t\big)\sqrt{x\delta t} = -\frac{k_{y}\theta_{y}\delta t - \big[1 - \big(k_{y}+\mu\xi_{y}\mp\eta m\xi_{y}^{2}\Delta\delta t\big)\delta t\big]x}{\xi_{y}\sqrt{x\delta t}}\hspace{1pt}.
\end{equation}
Clearly $z'_{2} \leq z'_{1}$, and suppose that $\delta_{T}\leq\frac{\sqrt{2}-1}{\sqrt{2}}\max\big\{0^{+},k_{y}+\mu\xi_{y}+\eta\xi_{y}^{2}\Delta T\big\}^{-1}$. Then we can find $z \in \left[z'_{2},z'_{1}\right]$ such that
\begin{equation}\label{eq3.112}
\Phi(z'_{1}) - \Phi(z'_{2}) = \left(z'_{1}-z'_{2}\right)\phi(z).
\end{equation}
First, if $x \leq k_{y}\theta_{y}\delta t\big[1 - \big(k_{y}+\mu\xi_{y}+\eta\xi_{y}^{2}\Delta T\big)\delta t\big]^{-1}$, we deduce from \eqref{eq3.111} and \eqref{eq3.112} that
\begin{equation}\label{eq3.113}
\Phi(z'_{1}) - \Phi(z'_{2}) \leq 2\eta m\nu_{y}\xi_{y}\Delta(\delta t)^{2}.
\end{equation}
Second, if $x > k_{y}\theta_{y}\delta t\big[1 - \big(k_{y}+\mu\xi_{y}+\eta\xi_{y}^{2}\Delta T\big)\delta t\big]^{-1}$, then $0 < z'_{2} \leq z'_{1}$ and hence
\begin{equation}\label{eq3.114}
\Phi(z'_{1}) - \Phi(z'_{2}) \leq \sqrt{\frac{2}{\pi}}\hspace{1pt}\eta m\xi_{y}\Delta(\delta t)^{3/2}h(x),
\end{equation}
where $h:(0,\infty)\mapsto\RR$ is defined by
\begin{equation}\label{eq3.115}
h(x) = \sqrt{x}\hspace{1pt}\exp\left\{-\frac{\big[k_{y}\theta_{y}\delta t - \big[1 - \big(k_{y}+\mu\xi_{y}+\eta\xi_{y}^{2}\Delta T\big)\delta t\big]x\big]^2}{2\hspace{.5pt}\xi_{y}^{2}x\delta t}\right\}.
\end{equation}
Proceeding as before, we find
\begin{equation}\label{eq3.116}
h(x) < \nu_{y}\sqrt{2\pi\delta t}\hspace{1pt},
\end{equation}
which again leads to the upper bound in \eqref{eq3.113}. Substituting back into \eqref{eq3.110}, we get
\begin{equation}\label{eq3.117}
\mathcal{I} \leq \exp\bigg\{\eta m\big(k_{y}\theta_{y}+2\nu_{y}\xi_{y}\big)\Delta(\delta t)^{2} + \frac{1}{2}\hspace{1pt}\mu^{2}x\delta t + b\Delta x\delta t\bigg\}.
\end{equation}
Combining \eqref{eq3.97}, \eqref{eq3.106} and \eqref{eq3.117}, we deduce that independent of the sign of $\tilde{x}$,
\begin{equation}\label{eq3.118}
\mathcal{I} \leq \exp\bigg\{\eta m\big(k_{y}\theta_{y}+2\nu_{y}\xi_{y}\big)\Delta(\delta t)^{2} + \frac{1}{2}\hspace{1pt}\mu^{2}x\delta t + c\Delta x\delta t\bigg\},
\end{equation}
where
\begin{equation}\label{eq3.119}
c = \eta m\big[1 - (k_{y}-|\mu|\xi_{y})\delta t\big] + \frac{1}{2}\hspace{1pt}\eta^{2}m^{2}\xi_{y}^{2}\Delta(\delta t)^{2}.
\end{equation}
Applying \eqref{eq3.93} with $\omega=\frac{m}{N}$ leads to
\begin{equation}\label{eq3.120}
\eta^{2}m^{2}\xi_{y}^{2}\Delta(\delta t)^{2} - 2\eta m(k_{y}-|\mu|\xi_{y})\delta t - 2\eta + 2 \leq 0\hspace{.5pt},
\end{equation}
and so
\begin{equation}\label{eq3.121}
c \leq \eta(m+1) - 1.
\end{equation}
Hence,
\begin{equation}\label{eq3.122}
\mathcal{I} \leq \exp\bigg\{\eta m\big(k_{y}\theta_{y}+2\nu_{y}\xi_{y}\big)\Delta(\delta t)^{2} - \lambda x\delta t + \eta(m+1)\Delta x\delta t\bigg\}.
\end{equation}
Substituting back into \eqref{eq3.95} with this upper bound gives the inductive step. Taking $m=N$ in \eqref{eq3.94},
\begin{equation}\label{eq3.123}
\E\big[\hspace{1pt}\oversymb{\hspace{-1pt}\Theta\hspace{-1pt}}\hspace{1pt}_{T}\big] < \exp\left\{\frac{1}{2}\hspace{1pt}\eta\Delta T^{2}\big(k_{y}\theta_{y}+2\nu_{y}\xi_{y}\big) + \eta\Delta Ty_{0}\right\}.
\end{equation}
The right-hand side is finite and independent of $\delta t$, which concludes the proof.
\end{proof}

Second, consider the symmetrized Euler scheme in \eqref{eq2.5} and let $\hspace{-.5pt}\oversymb{\hspace{.5pt}Y}_{\hspace{-2.5pt}t}=\tilde{y}_{t_{n}}$, $\forall\hspace{.5pt}t\in[t_{n},t_{n+1})$.

\begin{proposition}\label{Prop3.9}
If $\Delta\leq0$ and $T\geq0$ or otherwise, if $\Delta>0$ and $T\leq T^{*}$, with $T^{*}$ from \eqref{eq3.60} -- \eqref{eq3.61}, then there exists $\delta_{T}>0$ so that for all $\delta t\in(0,\delta_{T})$, the first moment of the exponential functional in \eqref{eq3.2} of the symmetrized Euler scheme is uniformly bounded, i.e.,
\begin{equation}\label{eq3.124}
\sup_{\delta t\in(0,\delta_{T})}\hspace{1.5pt}\sup_{t\in[0,T]}\E\big[\hspace{1pt}\oversymb{\hspace{-1pt}\Theta\hspace{-1pt}}\hspace{1pt}_{t}\big]<\infty\hspace{.5pt}.
\end{equation}
\end{proposition}
\begin{proof}
We follow the argument of Proposition \ref{Prop3.8} closely and note that $\tilde{x}\geq0$.
\end{proof}

%%%%%%%%%%%%%%%%%%%%%%%%%%%%%%%%%%%%%%%%%%%%%%%%%%%%%%%%%%%%%%%%%%%%%%%%%%%%%%
%%% Section 4 %%%%%%%%%%%%%%%%%%%%%%%%%%%%%%%%%%%%%%%%%%%%%%%%%%%%%%%%%%%%%%%%%%%%%%
%%%%%%%%%%%%%%%%%%%%%%%%%%%%%%%%%%%%%%%%%%%%%%%%%%%%%%%%%%%%%%%%%%%%%%%%%%%%%%
\section{Moment stability in the Heston model}\label{sec:Heston}

In this section, we consider a filtered probability space $\left(\Omega,\mathcal{F},\{\mathcal{F}_t\}_{t\geq0},\mathbb{Q}\right)$ and suppose that the dynamics of the underlying process are governed by the Heston model under measure~$\mathbb{Q}$:
\begin{align}\label{eq4.1}
	\begin{dcases}
	dS_{t} = rS_{t}dt + \sqrt{v_{t}}\hspace{.5pt}S_{t}\hspace{1pt}dW^{s}_{t} \\[3pt]
	dv_{t} \hspace{1pt} = k(\theta-v_{t})dt + \xi\sqrt{v_{t}}\,dW^{v}_{t},
	\end{dcases}
\end{align}
where $W^{s}$ and $W^{v}$ are correlated Brownian motions with constant correlation $\rho\in(-1,1)$, $r$ is an arbitrary non-negative number and $k$, $\theta$, $\xi>0$ as before. We decompose $W^{s}$ and write it as a linear combination of independent Brownian motions $\tilde{W}^{s}$ and $W^{v}$. An application of It\^o's formula leads to
\begin{equation}\label{eq4.2}
S_{T} = S_{0}\exp\bigg\{rT - \frac{1}{2}\int_{0}^{T}{\hspace{-1pt}v_{t}\hspace{1pt}dt} + \sqrt{1-\rho^{2}}\int_{0}^{T}{\hspace{-2pt}\sqrt{v_{t}}\hspace{1.5pt}d\tilde{W}^{s}_{t}}  + \rho\int_{0}^{T}{\hspace{-2pt}\sqrt{v_{t}}\hspace{1.5pt}dW^{v}_{t}}\bigg\}\hspace{.5pt}.
\end{equation}
In particular, we are interested in the evaluation of moments $\E\left[S_{T}^{\omega}\right]$ for $\omega>1$. \citet{Andersen:2007} give several examples of fixed income securities with super-linear payoffs, whose risk-neutral valuation involves the calculation of the second moment. Hence, moment explosions may lead to infinite prices of derivatives. Moreover, establishing the existence of moments of order higher than one of the process and its approximation is also important for the convergence analysis.

Conditioning on the $\sigma$-algebra $\mathcal{G}^{v}_{T}=\sigma\hspace{-1pt}\left(W^{v}_t; \hspace{2pt} 0\leq t\leq T\right)$, we find
\begin{equation}\label{eq4.3}
\E\!\big[S_{T}^{\omega}\big] = S_{0}^{\omega}\E\bigg[\exp\bigg\{\omega rT + \bigg[\frac{1}{2}\hspace{1pt}\omega(\omega-1)-\frac{1}{2}\hspace{1pt}\omega^{2}\rho^{2}\bigg]\int_{0}^{T}{\hspace{-1pt}v_{t}\hspace{1pt}dt} + \omega\rho\int_{0}^{T}{\hspace{-2pt}\sqrt{v_{t}}\hspace{1.5pt}dW^{v}_{t}}\bigg\}\bigg].
\end{equation}
Define the \textit{explosion time} of the moment of order $\omega$ to be the first time beyond which the moment $\E\left[S_{T}^{\omega}\right]$ will cease to exist, i.e.,
\begin{equation}\label{eq4.4}
T^{*}(\omega) \equiv \sup\big\{t\geq0 :\hspace{2pt} \E\left[S_{t}^{\omega}\right]<\infty\big\}.
\end{equation}
If the moment does not explode in finite time, then $T^{*}(\omega)=\infty$.

For convenience, we fix $S_{0}=1$ and $r=0$. Proposition \ref{Prop3.2} derives sharp conditions on the finiteness of moments of the process, whereas Propositions \ref{Prop3.4} and \ref{Prop3.6} to \ref{Prop3.9} give lower bounds on the explosion times of moments of the different discretizations. For illustration, we plot in Figure~\ref{fig:1} the explosion time and the corresponding lower bounds with different schemes against the model parameters. Since
\begin{equation}\label{eq4.5}
\Delta = \frac{1}{2}\hspace{1pt}\omega(\omega-1)\hspace{.5pt},
\end{equation}
Propositions \ref{Prop3.2}, \ref{Prop3.4} and \ref{Prop3.6} to \ref{Prop3.9} ensure the finiteness of the first moment of the process and its discretizations for all $T$, i.e., $T^{*}(1)=\infty$. On the other hand, we infer from Figure \ref{fig:1a} that both the explosion time for the exact process and the lower bounds with the explicit Euler discretizations approach infinity as $\omega$ approaches one, i.e., that $\lim_{\omega\to1^{+}}T^{*}(\omega)=\infty$. This ensures the uniform boundedness of moments, for $\omega$ sufficiently close to one, of the explicit schemes even for very long maturities, an important ingredient in proving the strong convergence of the approximation process \citep[see][]{Cozma:2015b}. Note that the green and the yellow curves in Figure \ref{fig:1} overlap when $\rho\geq0$.
\begin{figure}[htb]
\centering
\begin{subfigure}{.5\textwidth}
  \centering
  \includegraphics[width=.97\linewidth,height=1.7in]{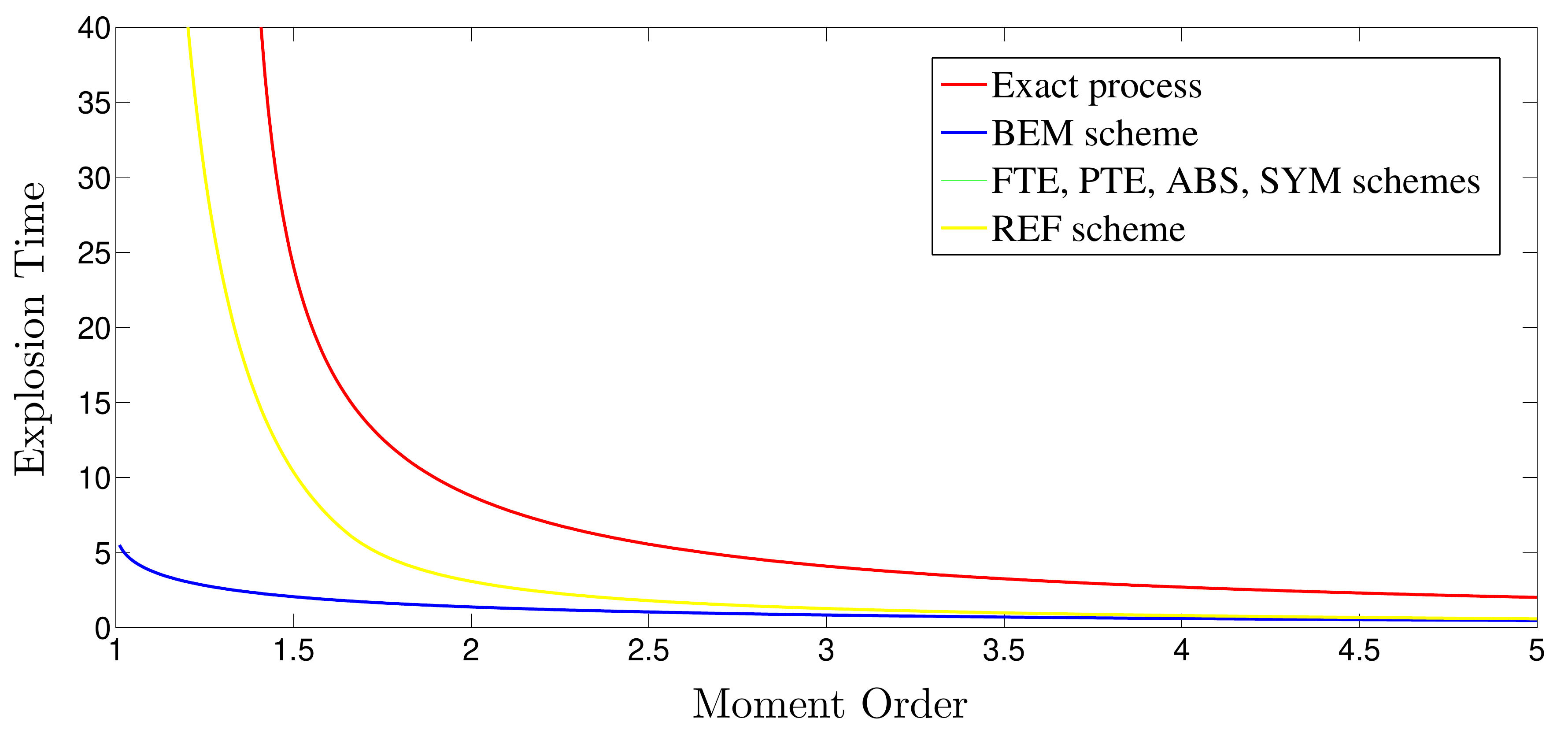}
  \caption{Against the moment order $\omega$.}
  \label{fig:1a}
\end{subfigure}%
\begin{subfigure}{.5\textwidth}
  \centering
  \includegraphics[width=.97\linewidth,height=1.7in]{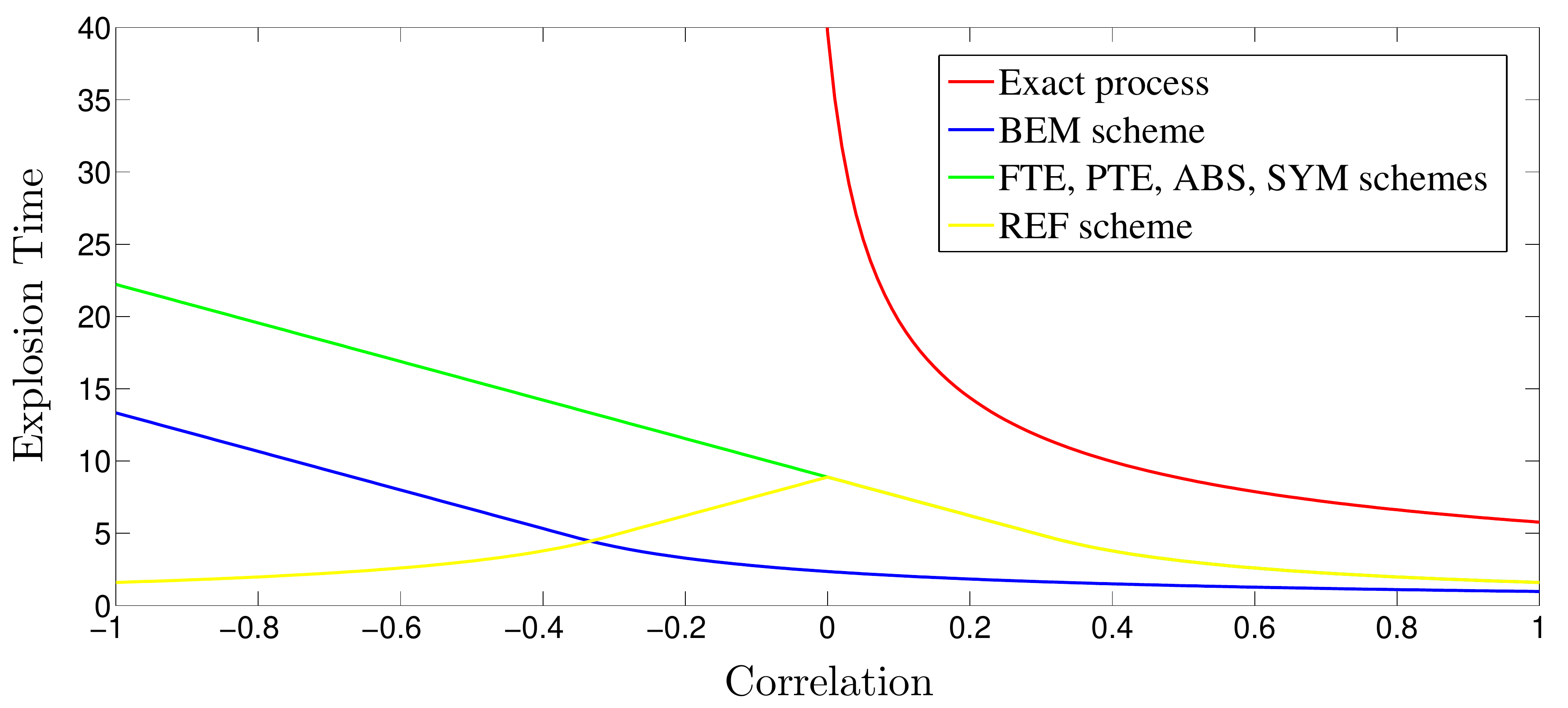}
  \caption{Against the correlation coefficient $\rho$.}
  \label{fig:1b}
\end{subfigure} \\[1em]
\begin{subfigure}{.5\textwidth}
  \centering
  \includegraphics[width=.97\linewidth,height=1.7in]{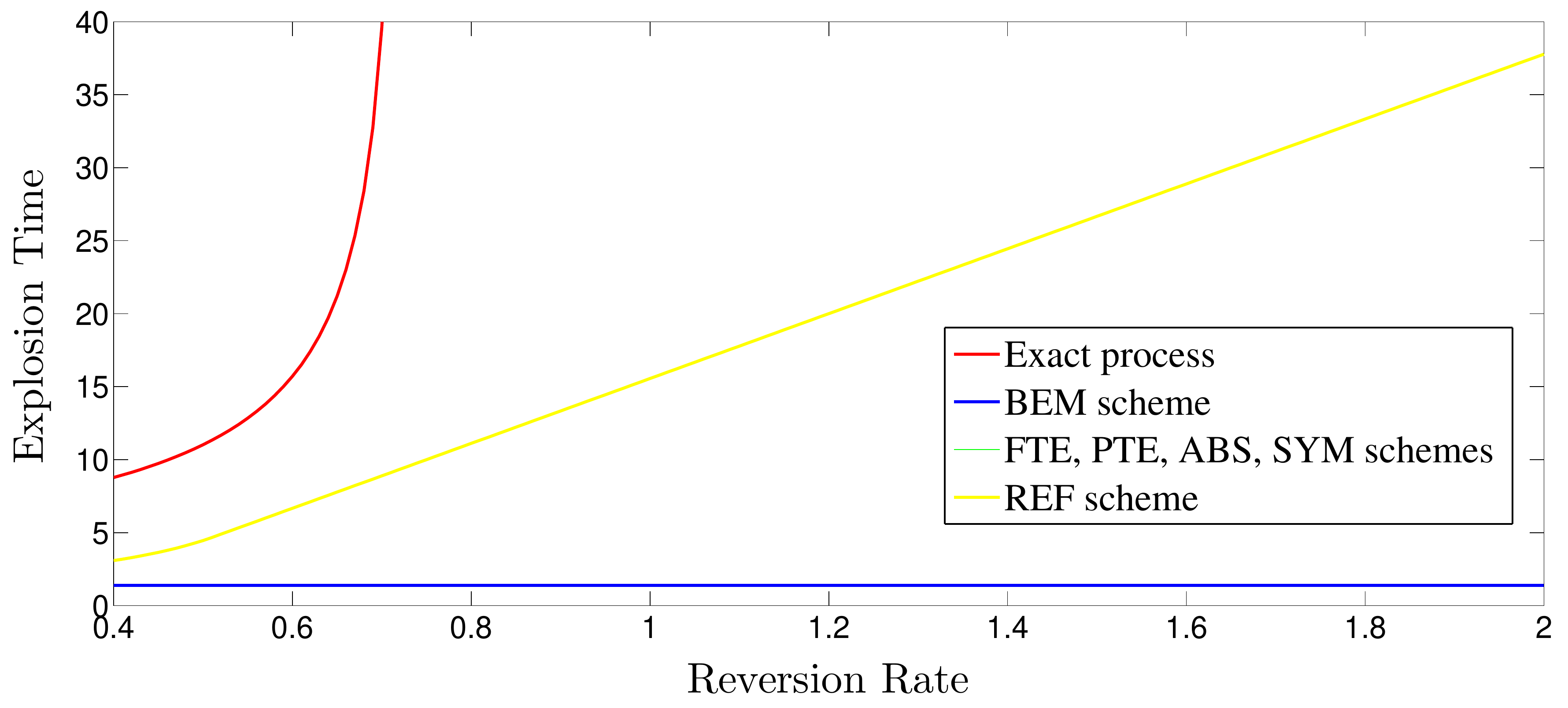}
  \caption{Against the mean reversion rate $k$.}
  \label{fig:1c}
\end{subfigure}%
\begin{subfigure}{.5\textwidth}
  \centering
  \includegraphics[width=.97\linewidth,height=1.7in]{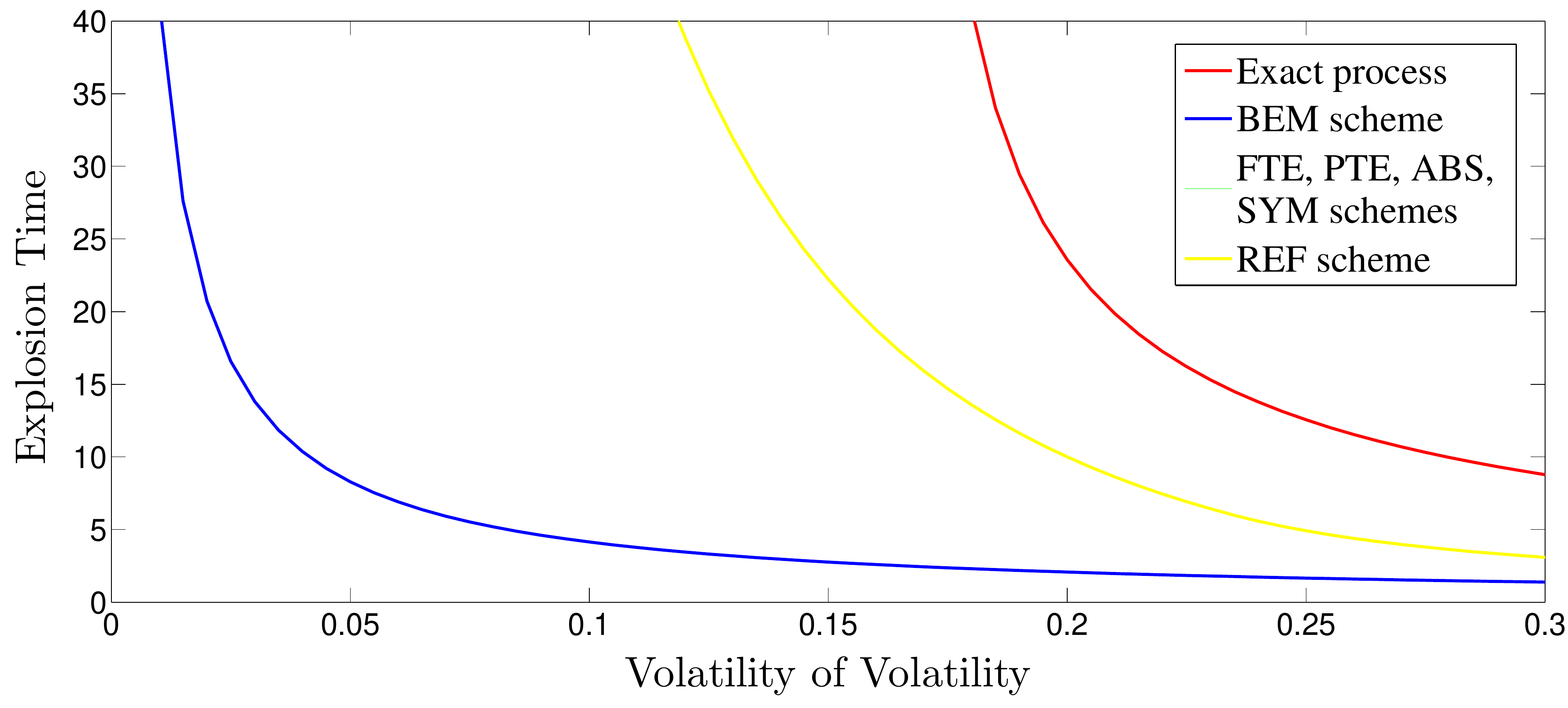}
  \caption{Against the volatility of volatility $\xi$.}
  \label{fig:1d}
\end{subfigure}%
\caption{The explosion time of moments of the exact process and the lower bounds with different discretization schemes, plotted against the model parameters when $k=0.4$, $\xi=0.3$, $\rho=0.5$ and $\omega=2$, except for the one which is varied.}
\label{fig:1}
\end{figure}

The data in Figure \ref{fig:1b} suggest that there exists a critical correlation level $\rho^{*}$ such that $\E\left[S_{T}^{\omega}\right]<\infty$ for all $T$, provided $\rho\leq\rho^{*}$, and $\E\left[S_{T}^{\omega}\right]=\infty$ for some $T$, provided $\rho>\rho^{*}$. When $k=0.4$, $\xi=0.3$ and $\omega=2$, we find $\rho^{*}=-0.04$. Moreover, we also infer from the data in Figure \ref{fig:1b} that decreasing the correlation has a damping effect on the second moment of the process, and that for strongly negative correlations between the underlying process and the variance, as is usually the case in equity markets (the so-called \textit{leverage effect}),~the~lower bounds on the explosion time of the second moment -- with all but the reflection scheme -- are above the typical maturity range of equity derivatives.

The data in Figures \ref{fig:1c} and \ref{fig:1d} indicate that increasing the speed of mean reversion and decreasing the volatility of volatility have a damping effect on the second moment of the process and its explicit Euler discretizations, which is to be expected considering that larger values of $k$ and smaller values of $\xi$ lead to smaller fluctuations in the variance over time. Next, we assign the following values to the underlying model parameters: $k=0.4$, $\theta=0.12$, $\xi=0.3$, $v_{0}=0.12$ and $\rho=1$, and note that the Feller condition is satisfied. Henceforth, we estimate the second moment by a standard Monte Carlo estimator.

From \eqref{eq3.10}, the explosion time of the second moment of the process is: $T^{*}=5.77$. On the other hand, we infer from the data in Figure \ref{fig:2} that for sufficiently small values of the time step -- for instance, when $\delta t=0.02$ -- the second moment with the BEM scheme will cease to exist after some time $T^{\text{\scalebox{0.75}[0.75]{BEM}}}$ close to $T^{*}$. The first sign of moment explosion can be observed when $T^{\text{\scalebox{0.75}[0.75]{BEM}}}=5.98$, however this phenomenon is more pronounced when $T=6.14$, where the moment jumps to $9.7\times10^{4}$.
\begin{figure}[htb]
\begin{center}
\includegraphics[width=1.0\columnwidth,keepaspectratio]{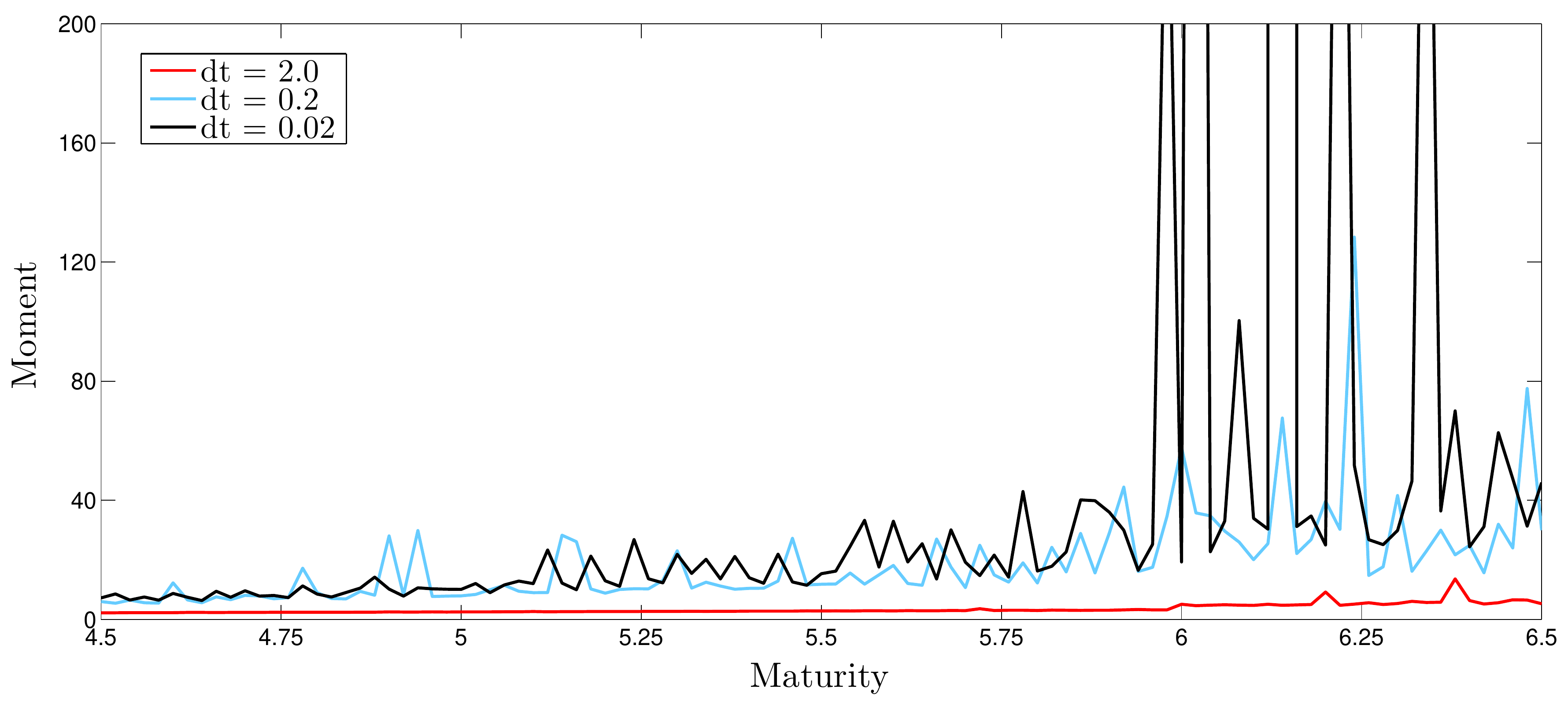}
\caption{The second moment of the Heston model, calculated using the BEM scheme and $3\times10^{7}$ simulations, plotted against the maturity $T$.}
\label{fig:2}
\end{center}
\end{figure}

By close inspection of the data in Figure \ref{fig:3}, we infer that for sufficiently small values of the time step $\delta t$, the approximation to the second moment with either the implicit or one of the explicit schemes considered in this paper explodes after some critical time close to $T^{*}$. Combined with the previous observation, this suggests that the explosion times of the second moments of the process and its discretizations become close as we increase the number of time steps. Therefore, we deduce from the data in Figure \ref{fig:1} that the lower bounds for the partial truncation, the full truncation, the absorption and the symmetrized Euler schemes are sharper than the ones for the drift-implicit and the reflection schemes.
\begin{figure}[htb]
\begin{center}
\includegraphics[width=1.0\columnwidth,keepaspectratio]{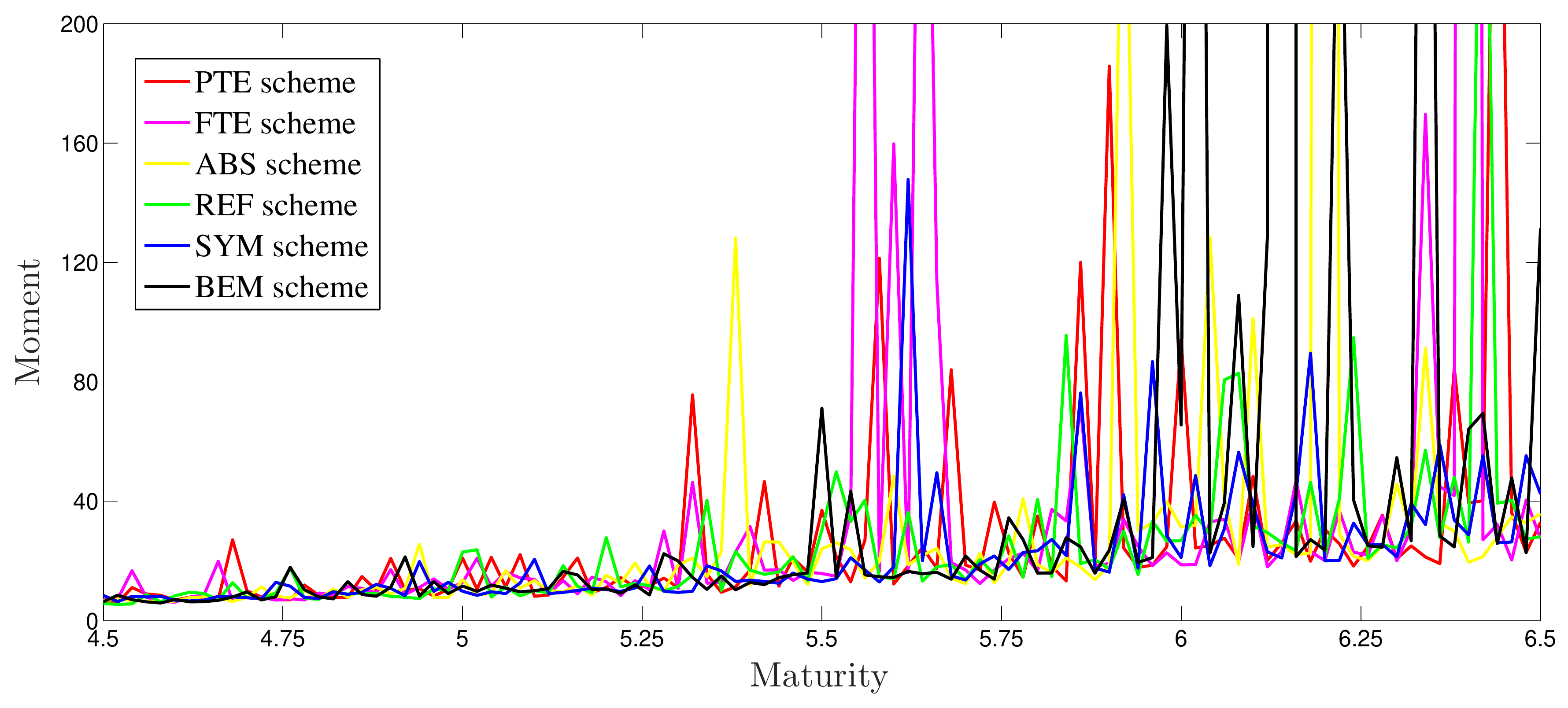}
\caption{The second moment of the Heston model, calculated with different discretization schemes, $\delta t=0.02$ and $3\times10^{7}$ simulations, plotted against the maturity $T$.}
\label{fig:3}
\end{center}
\end{figure}

%%%%%%%%%%%%%%%%%%%%%%%%%%%%%%%%%%%%%%%%%%%%%%%%%%%%%%%%%%%%%%%%%%%%%%%%%%%%%%
%%% Section 5 %%%%%%%%%%%%%%%%%%%%%%%%%%%%%%%%%%%%%%%%%%%%%%%%%%%%%%%%%%%%%%%%%%%%%%
%%%%%%%%%%%%%%%%%%%%%%%%%%%%%%%%%%%%%%%%%%%%%%%%%%%%%%%%%%%%%%%%%%%%%%%%%%%%%%
\section{Conclusions}\label{sec:conclusion}

In this paper, we have established the uniform exponential integrability of functionals of the Cox-Ingersoll-Ross process and a number of its discretization schemes often encountered in the finance literature. One consequence of this result with obvious practical implications is the stability of moments of numerical approximations for a large class of SDEs arising in finance, which in turn is used to prove strong convergence \citep{Higham:2002, Cozma:2015b}. An open question is whether we can find sharp conditions on the~exponential integrability of Euler approximations for the CIR process.
% In our future research, we intend to pursue extensions to models other than the one considered in this paper, including the CKLS model \citep{Chan:1992}.

%%%%%%%%%%%%%%%%%%%%%%%%%%%%%%%%%%%%%%%%%%%%%%%%%%%%%%%%%%%%%%%%%%%%%%%%%%%%%%
%%% Bibliography %%%%%%%%%%%%%%%%%%%%%%%%%%%%%%%%%%%%%%%%%%%%%%%%%%%%%%%%%%%%%%%%%%%%
%%%%%%%%%%%%%%%%%%%%%%%%%%%%%%%%%%%%%%%%%%%%%%%%%%%%%%%%%%%%%%%%%%%%%%%%%%%%%%
%\nocite{*}
\bibliographystyle{apa}
\bibliography{references}

\begin{comment}
%%%%%%%%%%%%%%%%%%%%%%%%%%%%%%%%%%%%%%%%%%%%%%%%%%%%%%%%%%%%%%%%%%%%%%%%%%%%%%
%%% Appendix %%%%%%%%%%%%%%%%%%%%%%%%%%%%%%%%%%%%%%%%%%%%%%%%%%%%%%%%%%%%%%%%%%%%%%
%%%%%%%%%%%%%%%%%%%%%%%%%%%%%%%%%%%%%%%%%%%%%%%%%%%%%%%%%%%%%%%%%%%%%%%%%%%%%%
\titleformat{\section}{\large\bfseries}{\appendixname~\thesection .}{0.5em}{}
\begin{appendices}

%%%%%%%%%%%%%%%%%%%%%%%%%%%%%%%%%%%%%%%%%%%%%%%%%%%%%%%%%%%%%%%%%%%%%%%%%%%%%%
%%% Section A %%%%%%%%%%%%%%%%%%%%%%%%%%%%%%%%%%%%%%%%%%%%%%%%%%%%%%%%%%%%%%%%%%%%%%
%%%%%%%%%%%%%%%%%%%%%%%%%%%%%%%%%%%%%%%%%%%%%%%%%%%%%%%%%%%%%%%%%%%%%%%%%%%%%%
\section{}\label{sec:}

\end{appendices}
\end{comment}
%%%%%%%%%%%%%%%%%%%%%%%%%%%%%%%%%%%%%%%%%%%%%%%%%%%%%%%%%%%%%%%%%%%%%%%%%%%%%%
%%% End of Document %%%%%%%%%%%%%%%%%%%%%%%%%%%%%%%%%%%%%%%%%%%%%%%%%%%%%%%%%%%%%%%%%%
%%%%%%%%%%%%%%%%%%%%%%%%%%%%%%%%%%%%%%%%%%%%%%%%%%%%%%%%%%%%%%%%%%%%%%%%%%%%%%
\end{document}